\documentclass{jpsj2}
\usepackage{amscd}
\usepackage{amsfonts}
\usepackage{amsmath}
\usepackage{amssymb}
\usepackage{amsthm}
\usepackage{color}

\theoremstyle{definition}
\newtheorem{theorem}{Theorem}[section]
\newtheorem{lemma}[theorem]{Lemma}
\newtheorem{coro}[theorem]{Corollary}

\newtheorem{proposition}[theorem]{Proposition}
\newtheorem{definition}[theorem]{Definition}
\newtheorem{example}[theorem]{Example}
\newtheorem{remark}[theorem]{Remark}

\numberwithin{equation}{section}

\newcommand\ul{\underline}
\newcommand\ot{\otimes}

\newcommand\Z{\mathbb{Z}}
\newcommand\Zn{\Z_{\ge0}}
\newcommand\wt{\mbox{\textsl{wt}}\,}
\newcommand\Wt{\mbox{\textsl{Wt}}\,}
\newcommand\otb{\stackrel{\bullet}{\otimes}}
\newcommand\otv{\stackrel{\vee}{\otimes}}
\newcommand\otn{\stackrel{\nabla}{\otimes}}

\def\ot{\otimes}

\newcommand{\batten}[4]{%
\begin{picture}(40,40)(-20,-20)
	\put(-10,0){\line(1,0){20}}
	\thicklines
	\put(0,10){\line(0,-1){20}}
	\put(-11,0){\makebox(0,0)[r]{$#1$}}
	\put(0,11){\makebox(0,0)[b]{$#2$}}
	\put(0,-11){\makebox(0,0)[t]{$#3$}}
	\put(11,0){\makebox(0,0)[l]{$#4$}}
\end{picture}
}
\newcommand{\oddseqa}[7]{%
\put(1,{#7}){\makebox(1,1){${#1}$}}
\put(3,{#7}){\makebox(1,1){${#2}$}}
\put(5,{#7}){\makebox(1,1){${#3}$}}
\put(7,{#7}){\makebox(1,1){${#4}$}}
\put(9,{#7}){\makebox(1,1){${#5}$}}
\put(11,{#7}){\makebox(1,1){${#6}$}}
}
\newcommand{\oddseqb}[7]{%
\put(13,{#7}){\makebox(1,1){${#1}$}}
\put(15,{#7}){\makebox(1,1){${#2}$}}
\put(17,{#7}){\makebox(1,1){${#3}$}}
\put(19,{#7}){\makebox(1,1){${#4}$}}
\put(21,{#7}){\makebox(1,1){${#5}$}}
\put(23,{#7}){\makebox(1,1){${#6}$}}
}
\newcommand{\evenseqa}[8]{%
\put(0,{#8}){\makebox(1,1){${#1}$}}
\put(2,{#8}){\makebox(1,1){${#2}$}}
\put(4,{#8}){\makebox(1,1){${#3}$}}
\put(6,{#8}){\makebox(1,1){${#4}$}}
\put(8,{#8}){\makebox(1,1){${#5}$}}
\put(10,{#8}){\makebox(1,1){${#6}$}}
\put(12,{#8}){\makebox(1,1){${#7}$}}
}
\newcommand{\evenseqb}[7]{%
\put(14,{#7}){\makebox(1,1){${#1}$}}
\put(16,{#7}){\makebox(1,1){${#2}$}}
\put(18,{#7}){\makebox(1,1){${#3}$}}
\put(20,{#7}){\makebox(1,1){${#4}$}}
\put(22,{#7}){\makebox(1,1){${#5}$}}
\put(24,{#7}){\makebox(1,1){${#6}$}}
}

\newcommand{\latticefig}{%
  \setlength{\unitlength}{6mm}
  \begin{picture}(25,11)(0,0)
\multiput(1,1.5)(2,0){12}{\line(1,0){1}}
\multiput(1.5,1)(2,0){12}{\line(0,1){1}}
\multiput(1,3.5)(2,0){12}{\line(1,0){1}}
\multiput(1.5,3)(2,0){12}{\line(0,1){1}}
\multiput(1,5.5)(2,0){12}{\line(1,0){1}}
\multiput(1.5,5)(2,0){12}{\line(0,1){1}}
\multiput(1,7.5)(2,0){12}{\line(1,0){1}}
\multiput(1.5,7)(2,0){12}{\line(0,1){1}}
\multiput(1,9.5)(2,0){12}{\line(1,0){1}}
\multiput(1.5,9)(2,0){12}{\line(0,1){1}}
\oddseqa{1}{5}{(-5}{3)}{(-5}{5)}{10}
\oddseqb{-5}{(-3}{1)}{(-1}{3)}{1}{10}
\evenseqa{0}{0}{4}{-4}{4}{-4}{4}{9}
\evenseqb{-4}{-2}{2}{-2}{2}{0}{9}
\oddseqa{1}{1}{3}{(-5}{3)}{(-3}{8}
\oddseqb{3)}{-5}{(-3}{3)}{(-1}{3)}{8}
\evenseqa{2}{0}{0}{2}{-4}{4}{-4}{7}
\evenseqb{4}{-4}{-2}{2}{-2}{2}{7}
\oddseqa{3)}{1}{1}{1}{(-5}{5)}{6}
\oddseqb{(-5}{3)}{-5}{(-1}{3)}{(-1}{6}
\evenseqa{-2}{2}{0}{0}{0}{-4}{4}{5}
\evenseqb{-4}{4}{-4}{0}{2}{-2}{5}
\oddseqa{(-1}{3)}{1}{1}{-1}{(-3}{4}
\oddseqb{3)}{(-5}{3)}{(-5}{1)}{3}{4}
\evenseqa{2}{-2}{2}{0}{0}{0}{-2}{3}
\evenseqb{2}{-4}{4}{-4}{2}{2}{3}
\put(1,2){\makebox(1,1){$3$}}
\put(3,2){\makebox(1,1){$(-1$}}
\put(5,2){\makebox(1,1){$3)$}}
\put(7,2){\makebox(1,1){$1$}}
\put(9,2){\makebox(1,1){$-1$}}
\put(11,2){\makebox(1,1){$-1$}}
\put(13,2){\makebox(1,1){$(-1$}}
\put(15,2){\makebox(1,1){$1)$}}
\put(17,2){\makebox(1,1){$(-5$}}
\put(19,2){\makebox(1,1){$3)$}}
\put(21,2){\makebox(1,1){$(-5$}}
\put(23,2){\makebox(1,1){$3)$}}
\put(0,1){\makebox(1,1){$4$}}
\put(2,1){\makebox(1,1){$2$}}
\put(4,1){\makebox(1,1){$-2$}}
\put(6,1){\makebox(1,1){$2$}}
\put(8,1){\makebox(1,1){$0$}}
\put(10,1){\makebox(1,1){$0$}}
\put(12,1){\makebox(1,1){$0$}}
\put(14,1){\makebox(1,1){$0$}}
\put(16,1){\makebox(1,1){$0$}}
\put(18,1){\makebox(1,1){$-4$}}
\put(20,1){\makebox(1,1){$4$}}
\put(22,1){\makebox(1,1){$-4$}}
\put(24,1){\makebox(1,1){$4$}}
\oddseqa{5)}{3}{(-1}{3)}{-1}{-1}{0}
\oddseqb{(-1}{1)}{-1}{(-5}{3)}{(-5}{0}
  \end{picture}
}

\title{Creation of Ballot Sequences in a Periodic Cellular Automaton}
\author{Taichiro Takagi}
\inst{Department of Applied Physics, National Defense Academy, Kanagawa 239-8686, Japan}
\abst{Motivated by an attempt to develop
a method for solving initial value problems in
a class of one dimensional periodic
cellular automata (CA) associated with crystal bases and soliton equations,
we consider a generalization of
a simple proposition in elementary mathematics.
The original proposition says that
any sequence of letters $1$ and $2$, 
having no less $1$'s than $2$'s, 
can be changed into a ballot sequence via
cyclic shifts only.
We generalize it to treat sequences of cells of common capacity $s > 1$,
each of them containing consecutive $2$'s (left) and $1$'s (right),
and show that these sequences can be changed into a ballot sequence via
two manipulations,
cyclic and ``quasi-cyclic" shifts. 
The latter is a new CA rule and we find that
various
kink-like structures are traveling along the system like particles
under the time evolution of this rule.}

\kword{discrete integrable systems, crystal bases, periodic box-ball system, combinatorics, Yang-Baxter relation}

\begin{document}
\maketitle
\section{Introduction}\label{sec:1}
It is well-known that the soliton equations provide 
infinite dimensional
completely integrable systems in the sense
that they possess infinitely many first integrals \cite{A,AC}.
The Korteweg--de Vries (KdV) equation 
is a representative example.
It is related to the Toda equation for nonlinear lattice dynamics \cite{Toda} and
the Lotka--Volterra equation for a model of a local ecosystem.
In 1990's
a cellular automaton (CA) now known as the box-ball system was invented by
Takahashi and Satsuma \cite{Th,TS}, which turned out to be a certain discrete limit
(ultradiscrete(UD) limit) of these equations \cite{TTMS}.
This CA indeed pickes up 
one of the specific features of the KdV equation, the existence of soliton solutions.
Several years later
generalizations of this CA were studied \cite{HIK,FOY,HHIKTT,Takagi} from the point of view of
models in statistical mechanics \cite{B} and the theory of crystal bases 
\cite{K}.

The box-ball system with periodic boundary conditions (pBBS)
appeared later as yet another CA associated with soliton equations \cite{YYT,MIT,KT1,KT2,KS1,IT}.
Because of the compactness of its level set
it has more profound mathematical structures than its non-periodic counterpart. 
As in the non-periodic case
this CA has commutative time evolution operators $T_l (l=1,2,...)$ 
\cite{KTT}.
In an analogy with Hamiltonian mechanics we regard
these time evolutions as
flows on a Liouville torus for some Hamiltonian vector fields 
associated with first integrals \cite{A}.

In this paper we study an aspect of a generaliztion of pBBS \cite{KS}.
In the analogy mentioned above
we pay our attention to a special two-dimensional sub-torus
in the whole Liouville torus and study its combinatorial property.
Our main result is Theorem \ref{th:may21_1}.
This is
motivated by an attempt to develop
a method for solving initial value problems in the generalized pBBS.
For this application and
a study of this CA related to combinatorial
Bethe ansatz \cite{KKR,KR}, see ref.~\citen{KS}.

The purpose of this paper is twofold.
One is to prove this theorem itself and the other
is to illustrate a new simple CA rule ``quasi-cyclic shift"
and
study its various properties through the proof of the theorem.
Despite the simplicity of its definition
this CA rule shows up many non-trivial phenomena which make it
worth exploring in its own right.
For instance, various
kink-like structures are traveling along the system like particles
under the time evolution of this rule.

We explain the problem more precisely.
Let $s>1$ be an integer that labels the generalized pBBS.
A commutative family of
time evolutions $T_l$s is 
defined by crystals for $l$-fold symmetric tensor representation
of quantum enveloping algebra $U_q(\widehat{sl}_2)$ \cite{KTT}.
Then the time evolution $T_s$ is merely a cyclic shift.
We pay our attention to the next simple time evolution $T_{s-1}$
and call it the quasi-cyclic shift.
For example, the quasi-cyclic shift of the odd integer sequence
$(1,5,-5,3,-5,5,-5,-3,1,-1,3,1)$ is
$(1,1,3,-5,3,-3,3,-5,-3,3,-1,3)$ (See Fig.~\ref{fig:1}).
We interpret these integers as weights of crystals
for $s$-fold symmetric tensor representation of $U_q(\widehat{sl}_2)$.
The total weight, sums of these integers,
is invariant under the time evolutions.
Recall that there are special elements of crystals known as the highest weight elements.
Our main theorem means that if the total weight is non-negative then
the 2-dimensional sub-torus generated by the cyclic and the quasi-cyclic shifts contains
at least one highest weight element.

By representing the elements of crystals as words, 
the highest weight elements become ballot sequences.
Consider a sequence of letters 1 and 2.
Let $\pi = a_1 \ldots a_L$ be such a sequence and
$\pi_k = a_1 \ldots a_k \, (1 \leq k \leq L)$ its prefixes.
If the number of 1's in any $\pi_k$ is not less than that of 2's,
then $\pi$ is called a ballot sequence, or lattice permutation \cite{S,Sa}.
A simple proposition in elementary mathematics says that
any sequence of letters $1$ and $2$, having no less $1$'s than $2$'s, 
can be changed into a ballot sequence via cyclic shifts only.
Our main theorem can be interpreted as a generalization of this proposition.
We treat sequences of cells of common capacity $s > 1$,
each of them containing consecutive $2$'s (left) and $1$'s (right),
and show that these sequences can be changed into a ballot sequence via
only two manipulations mentioned in the previous paragraph. 
For later reference,
we call it \textit{the problem of creating ballot sequences}
or \textit{the ballot sequence problem} for short,
that asks whether this claim holds.

The rest of this paper is organized as follows.
In \S \ref{sec:2} we review a formulation of the
integrable cellular automata and find the problem of 
creating ballot sequences.
In \ref{subsec:2_1} we introduce the notion of crystals
and fix their notations.
The crystal of $s$-fold symmetric tensor representation
of $U_q(\widehat{sl}_2)$ is denoted by $B_s$.
The elements of $B_s$ are expressed in two different ways,
by means of \textit{word} and \textit{weight}.
A formula for the combinatorial $R$ between
$B_{s-1} \ot B_s$ and $B_s \ot B_{s-1}$
in the latter expression is given
in Lemma \ref{lem:july8_1},
which will be used throughout this paper.
Based on refs.~\citen{KS} and \citen{T}
we review the definition of the time evolution operators of
the cellular automata
and the notion of the evolvability 
of the paths in \ref{subsec:2_2}.
In \ref{subsec:2_3} we formulate the ballot sequence problem and
present the main result (Theorem \ref{th:may21_1}).

In \S \ref{sec:3} we begin by developing necessary tools to
give a proof of the theorem in the case of odd cell capacity.
In \ref{subsec:3_1} we introduce the notions of
local minimum, (virtual) global minimum, and modified weight.
We characterize a ballot sequence by
the positivity of the modified weight.
The relevance of a special type of local minimum, called a spike,
to the ballot sequence problem is also shown.
In \ref{subsec:3_2} we prove that the number of
local minima in a path is preserved by $T_{s-1}$.
In \ref{subsec:3_3} we show that every path can be uniquely decomposed into 
certain characteristic \textit{segments}.
The notions of \textit{maximal} and \textit{singular} segments
are introduced and their conservation 
under the time evolution is proved.
While the conservation of the former ensures the persistency of
the evolvability of a path, that of the latter
suggests the existence of two different types of
local minima.

In \ref{subsec:3_4} the local minima are classified into 
the two types, called \textit{fixed} and \textit{floating} minima.
We also discuss the method to attach labels to the local minima.
This enables us to regard the local minima 
as distinguishable particles moving along a path.
In \ref{subsec:3_5} we show that the difference of the modified weights
of any two fixed minima is preserved.
In \ref{subsec:3_6} we prove that every fixed minimum becomes
a spike on some occasion.

The main theorem is proved 
in \ref{subsec:3_7} in the case of zero total weight.
We introduce the notion of a lowest fixed minimum.
It turns out to be a global minimum, and remains so in any future time
due to the result in \ref{subsec:3_5}.
Then on some occasion it becomes a spike as shown in \ref{subsec:3_6},
which makes the path be equivalent to a ballot sequence.
In a similar way, the theorem is proved 
in \ref{subsec:3_8} in the case of positive total weight.

In \S \ref{sec:4} we treat the case of even cell capacity
and give a proof of the main theorem in this case.
A few discussions are gathered in \S \ref{sec:5} and a proof of a proposition is given in Appendix \ref{app:a}.
\section{Integrable Periodic Cellular Automata}
\label{sec:2}
\subsection{\mathversion{bold}Crystal bases and combinatorial $R$}
\label{subsec:2_1}
Let $B_s$ be the crystal of the $s$-fold symmetric tensor representation
of $U_q(A_1^{(1)})$.
In a conventional setting it is given by $B_s = 
\{ x=(x_1, x_2) \in (\Zn)^2 | x_1+x_2=s \}$ as a set,
equipped with some algebraic structures such as
raising/lowering (Kashiwara) operators and tensor products.
We begin by giving the other expressions for $B_s$ which are
suitable for our study.
The element $(x_1, x_2)$ can be expressed as
a parenthesized sequence of consecutive 2's (left) and 1's (right)
satisfying $\#(1)=x_1, \#(2)=x_2$.
For example $B_1 = \{(1),(2)\}, B_2 =\{(11),(21),(22) \}
, B_3 =\{(111),(211),(221),(222) \}$. 
We call it the \textbf{word notation}.
On the other hand,
another useful notation is defined as follows.
Given $x=(x_1, x_2)$ its weight is defined as
$\wt (x) = x_1 - x_2$.
With a fixed $s$ the elements of $B_s$ are determined by
their weights, hence we can use them as an expression for $x$.
For example $B_1 = \{1,-1\}, B_2 =\{2,0,-2 \},
B_3 = \{3,1,-1,-3 \}$. 
We call it the \textbf{weight notation}.
In this paper we use both notations:
The former is mainly used to exhibit
the results, while the latter to derive them.
Throughout this paper we shall omit the symbol $\ot$ for the
tensor product in the word notation, {and}
in some case we shall use $\dot{1},\dot{2},\ldots$ instead of
$-1,-2,\ldots$ in the weight notation for typographical reasons.

Given $r,s$ a pair of positive integers there exists a
unique isomorphism between the two tensor products $B_r \ot B_s$ and
$B_s \ot B_r$, that commutes with Kashiwara operators
up to a constant shift.
It is called the \textbf{\mathversion{bold}combinatorial $R$} and denoted by $R$.
It satisfies the Yang-Baxter relation
\begin{equation}\label{eq:july20_1}
(1 \ot R)(R \ot 1)(1 \ot R) =
(R \ot 1)(1 \ot R)(R \ot 1) \quad \mbox{on} \quad
B_s \ot B_r \ot B_t.
\end{equation}
This relation yields commuting family of
time evolution operators in the cellular automata that we are going to discuss.
More details about these objects are available
in subsection 2.1 of ref.~\citen{KTT}.
We need only the following formula.
In the weight notation the isomorphism
$R: B_r \ot B_s \rightarrow B_s \ot B_r$ is
given by $R: a \ot b \mapsto \tilde{b} \ot \tilde{a}$ {where}
\begin{align}
&\tilde{a} = a + \min(r-a,s+b)-\min(r+a,s-b),\nonumber\\
&\tilde{b} = b + \min(r+a,s-b)-\min(r-a,s+b). 
\label{eq:jun5_1}
\end{align}
The relation is depicted as
\begin{displaymath}
\batten{a}{b}{\tilde{b}}{\tilde{a}}.
\end{displaymath}
It is clear that if $r=s$ then $\tilde{a}=b,
\tilde{b}=a$.
For $r=s-1$ it is also easy to obtain the following result.
\begin{lemma}\label{lem:july8_1}
In the weight notation
the map $R: B_{s-1} \ot B_s \rightarrow B_s \ot B_{s-1}$ is
given by
$R: a \ot b \mapsto a+1 \ot b-1 \, \mbox{for $a+b>0$}, \quad
a-1 \ot b+1 \, \mbox{for $a+b<0$}$.
\begin{align*}
&\batten{a}{b}{a+1}{b-1} \qquad\qquad
\batten{a}{b}{a-1}{b+1} \\
& a+b>0 \qquad\qquad a+b<0
\end{align*}
\end{lemma}
We give an example of combinatorial $R$.
\begin{example}\label{ex:july30_1}
$R: B_1 \ot B_3 \to B_3 \ot B_1$ and
$R: B_2 \ot B_3 \to B_3 \ot B_2$ are
given as follows (in the word notation).
\begin{align*}
&
\batten{(1)}{(111)}{(111)}{(1)} \quad\qquad
\batten{(2)}{(111)}{(211)}{(1)} \quad\qquad
\batten{(11)}{(111)}{(111)}{(11)} \quad\qquad
\batten{(21)}{(111)}{(211)}{(11)} \quad\qquad
\batten{(22)}{(111)}{(221)}{(11)} 
\\
&
\batten{(1)}{(211)}{(111)}{(2)} \quad\qquad
\batten{(2)}{(211)}{(221)}{(1)} \quad\qquad
\batten{(11)}{(211)}{(111)}{(21)} \quad\qquad
\batten{(21)}{(211)}{(211)}{(21)} \quad\qquad
\batten{(22)}{(211)}{(222)}{(11)}
\\
&
\batten{(1)}{(221)}{(211)}{(2)} \quad\qquad
\batten{(2)}{(221)}{(222)}{(1)} \quad\qquad
\batten{(11)}{(221)}{(111)}{(22)} \quad\qquad
\batten{(21)}{(221)}{(221)}{(21)} \quad\qquad
\batten{(22)}{(221)}{(222)}{(21)}
\\
&
\batten{(1)}{(222)}{(221)}{(2)} \quad\qquad
\batten{(2)}{(222)}{(222)}{(2)} \quad\qquad
\batten{(11)}{(222)}{(211)}{(22)} \quad\qquad
\batten{(21)}{(222)}{(221)}{(22)} \quad\qquad
\batten{(22)}{(222)}{(222)}{(22)}
\end{align*}
We use this result in later examples.
\end{example}
\subsection{Paths and their time evolution}
\label{subsec:2_2}
Let $s$ and $L$ be positive integers.
The element of $(B_s)^{\ot L}$, the $L$-fold tensor product of $B_s$, 
is called a \textbf{path}.
We introduce the notion of
time evolution operators $T_1,T_2,\ldots$
of a path by means of the combinatorial $R$ 
\cite{KS}.
Then the path will be regarded as a state of a one dimensional
cellular automaton with a periodic boundary condition.
For instance, the original periodic box-ball system \cite{YYT} can be obtained
by setting $s=1$ and giving the time evolution as
$T_r \,(r \gg 1)$ which is defined by the
combinatorial $R$ (\ref{eq:jun5_1}) with $s=1$,
$r+a \to \infty$ ($r-a$ is fixed).

Given $v \in B_r, p \in (B_s)^{\ot L}$
define $v'=v'(v,p) \in B_r, p'=p'(v,p) \in (B_s)^{\ot L}$
by $v \ot p \mapsto p' \ot v'$, using 
the combinatorial $R$ repeatedly.
More precisely:
(1) We write $p$ as $p=b_1 \ot \cdots \ot b_L \, (b_i \in B_s)$
and set $v_0 = v$.
(2) Given $v_{i-1} \in B_r$ we define $b_i' \in B_s$ and
$v_i \in B_r$ by 
$R: v_{i-1} \ot b_i \mapsto b'_i \ot v_i \,(1 \leq i \leq L)$.
(3) We set $p'=b'_1 \ot \cdots \ot b'_L$ and $v'=v_L$.
The process is depicted by the 
\textbf{transition diagram}
\begin{equation}\label{eq:july23_3}
\batten{v}{b_1}{b_1'}{v_1}\!\!\!
\batten{}{b_2}{b_2'}{v_2}\!\!\!
\batten{}{}{}{\cdots\cdots}
\quad
\batten{}{}{}{v_{L-2}}\,\,
\batten{}{b_{L-1}}{b_{L-1}'}{v_{L-1}}\,\,
\batten{}{b_L}{b_L'}{v'.}
\end{equation}
\begin{definition}[ref.~\citen{KS}]
A path $p \in (B_s)^{\ot L}$ is \textbf{\mathversion{bold} $T_r$-evolvable}
if the following two conditions are satisfied.
\begin{enumerate}
\item There exists $v \in B_r$ such that $v'(v,p)=v$.
\item Suppose there exist 
$v,\tilde{v} \in B_r$ such that $v'(v,p)=v,v'(\tilde{v},p)=\tilde{v}$.
Then $p'(v,p)=p'(\tilde{v},p)$.
\end{enumerate}
\end{definition}
Suppose $p$ is $T_r$-evolvable.
Let $v$ be any element of $B_r$ such that $v'(v,p)=v$.
Then we write $T_r(p)$ for $p'(v,p)$.
This defines an invertible time evolution operator $T_r$, and
$T_r(p)$ is regarded as a path obtained from $p$
by applying $T_r$ to it.
It should be noted that such
$T_r(p)$ may not be $T_r$-evolvable any more,
so the existence of $T_r^2(p)$
is not always guaranteed.
\begin{example}\label{ex:july13_4}
Consider $p=(1111)(2222)(1111)(2222) \in (B_4)^{\ot 4}$ (word notation)
or $p = 4 \ot -4 \ot 4 \ot -4 \in (B_4)^{\ot 4}$ (weight notation).
Then $v=(222)$ is the unique element in $B_3$
which yields $v'(v,p) = v$.
So $p$ is $T_3$-evolvable.
We have
\begin{displaymath}
\batten{(222)}{(1111)}{(2221)}{(111)} \,\,\,\,
\batten{}{(2222)}{(2111)}{(222)}\,\,\,\,
\batten{}{(1111)}{(2221)}{(111)} \,\,\,\,
\batten{}{(2222)}{(2111)}{(222)}\quad,
\end{displaymath}
or
\begin{displaymath}
\batten{-3}{4}{-2}{3} \!\!\!\!
\batten{}{-4}{2}{-3}
\batten{}{4}{-2}{3} \!\!\!\!
\batten{}{-4}{2}{-3}\quad.
\end{displaymath}
Thus $T_3(p)=(2221)(2111)(2221)(2111)$ or
$T_3(p)= -2 \ot 2 \ot -2 \ot 2$.
\end{example}
\begin{example}\label{ex:july13_3}
Consider $\tilde{p}=(2221)(2111)(2221)(2111) \in (B_4)^{\ot 4}$
that is $T_3(p)$ in the previous example.
Then we have $v'((111),\tilde{p}) = (111)$, so Item 1 in the above definition
is satisfied by $v=(111) \in B_3$.
But we also have $v'((211),\tilde{p}) = (211)$ and
$p'((111),\tilde{p}) \ne p'((211),\tilde{p})$.
This is confirmed by the following diagrams.
\begin{displaymath}
\batten{3}{-2}{4}{-3}
\batten{}{2}{-4}{3} \!\!\!\!
\batten{}{-2}{4}{-3}
\batten{}{2}{-4}{3}
\quad
\batten{1}{-2}{0}{-1}
\batten{}{2}{0}{1} \!\!\!\!
\batten{}{-2}{0}{-1}
\batten{}{2}{0}{1}
\end{displaymath}
So Item 2 is not satisfied, hence $\tilde{p}$ is not $T_3$-evolvable.
\end{example}
We say that $p$ is \textbf{evolvable} if
it is $T_r$-evolvable for all $r \in \Zn$.
The following facts are known.
\begin{proposition}[ref.~\citen{KS}]
\label{prop:jun5_2}
A path $p = b_1 \ot \cdots \ot b_L
\in (B_s)^{\ot L}$ is evolvable if and only if
there exists $i$ such that $b_i=s$ or $b_i=-s$.
\end{proposition}
In particular every path in $(B_1)^{\ot L}$ is evolvable \cite{KTT}.
\begin{proposition}[ref.~\citen{KS}]\label{prop:july23_4}
Let $p$ be an evolvable path and
suppose both $T_t(p)$ and $T_r(p)$ are evolvable.
Then the commutativity $T_tT_r(p)=T_rT_t(p)$ holds.
\end{proposition}
This is a result of the Yang-Baxter relation (\ref{eq:july20_1}),
and ensures the integrability of the system.

Every path in $(B_s)^{\ot L}$ is $T_s$-evolvable. 
Here $T_s$ agrees with a right cyclic shift,
i.~e.~if $p = b_1 \ot b_2 \ot \cdots \ot b_L$ then
$T_s(p) = b_L \ot b_1 \ot \cdots \ot b_{L-1}$.

\begin{remark}\label{rem:jun30_1}
Let $p =b_1 \ot \cdots \ot b_L\in (B_s)^{\ot L}$ be an evolvable path.
If $r < s$ a practical method to obtain $T_r(p)$ is given as follows.
According to Proposition \ref{prop:jun5_2} 
there exists $i$ such that $b_i= \pm s$.
Let $\tilde{p} = (T_s)^{L-i}(p)= b_{i+1} \ot b_{i+2} \ot \cdots \ot b_i$.
{}From the explicit formula (\ref{eq:jun5_1}) 
we always have $\tilde{a} = \pm r$ 
under $R: a \ot b \mapsto \tilde{b} \ot \tilde{a}$ with $b=\pm s$.
Thus $v'(v,\tilde{p})=\pm r$ for any $v \in B_r$.
Define $T_r(\tilde{p})$ through the relation
$(\pm r) \ot \tilde{p} \mapsto T_r(\tilde{p}) \ot (\pm r)$.
Then $T_r(p)$ is given by $T_r(p)=(T_s)^{i} T_r(\tilde{p})$
since we have Proposition \ref{prop:july23_4}.
\end{remark}

\begin{example}\label{ex:jun30_2}
Let us consider the action of $T_1$ on the path
$p=(221)(111)(221)(221)$ which is evolvable.
By applying the cyclic shift twice, we obtain
\begin{displaymath}
\tilde{p} = T_3^2(p)= (221)(221)(221)(111).
\end{displaymath}
Then we have $v'(v,\tilde{p})=(1)$ for any $v \in B_1$.
By the combinatorial $R$ in Example \ref{ex:july30_1} we obtain
\begin{displaymath}
\batten{(1)}{(221)}{(211)}{}
\batten{(2)}{(221)}{(222)}{}
\batten{(1)}{(221)}{(211)}{}
\batten{(2)}{(111)}{(211)}{(1),}
\end{displaymath}
which leads to
$T_1(\tilde{p}) = (211)(222)(211)(211)$.
Hence 
\begin{displaymath}
T_1(p) = T_3^2 T_1(\tilde{p}) =(211)(211)(211)(222).
\end{displaymath}
\end{example}

Since the time evolution $T_{s-1}$ 
for $(B_s)^{\ot L}$ has a particular importance to our study,
we introduce a definition and a notation {specific} to it.
As we have seen in Examples \ref{ex:july13_4} and \ref{ex:july13_3}
a path of the form $p_{\pm} := (\pm s \ot \mp s )^{\ot L/2}$
(when $L$ is even) is $T_{s-1}$-evolvable but $T_{s-1}(p_{\pm})$ is not.
We will find that $T_{s-1}$ can be applied
arbitrary times to any evolvable path except
$p_{\pm}$ (Corollary \ref{cor:jun8_2}).
By this observation and
in view of Proposition \ref{prop:jun5_2}
we give the following definition.
\begin{definition}
A path $p = b_1 \ot \cdots \ot b_L
\in (B_s)^{\ot L}$ is 
\textbf{\mathversion{bold}
$T_{s-1}$-evolvable in a strong sense}
if the following two conditions are satisfied.
(1) There exists $i$ such that $b_i=s$ or $b_i=-s$.
(2) The path does not take the form $p_{\pm}$.
\end{definition}
We also define
\begin{displaymath}
\mathcal{P}_{L,s}=
\left\{ p = b_1 \ot \cdots \ot b_L \in (B_s)^{\ot L}
\,\vert \,\mbox{$p$ is $T_{s-1}$-evolvable in a strong sense.} \right\}
\end{displaymath}
for positive integers $L,s$.

\subsection{The problem of creating ballot sequences}
\label{subsec:2_3}
%
Given a path $p =b_1 \ot \cdots \ot b_L\in (B_s)^{\ot L}$ 
(weight notation) we 
define its \textbf{total weight} by $\wt (p) = \sum_{i=1}^L b_i$.
{}From now to the end of this subsection
we understand each of the paths to be
expressed in the word notation, and
we ignore all the parentheses to ``read" a path
as a single word that consists of the letters $1$ and $2$.
Recall the definition of ballot sequence in \S \ref{sec:1}.
We say that a path $p \in (B_s)^{\ot L}$ is
\textbf{equivalent to a ballot sequence}
if there exists $i$ such that $T_s^i(p)$ is a ballot sequence.

Let us begin by considering special cases.
The following fact is easy to prove.
\begin{proposition}
Every $p \in (B_1)^{\ot L}$ of nonnegative total weight
is equivalent to a ballot sequence.
\end{proposition}
This result was a key to the method in \cite{KTT}.
An analogue of this proposition is given as follows.
\begin{proposition}[ref.~\citen{T}]
Suppose $p \in (B_2)^{\ot L}$ is evolvable and $\wt (p) > 0$.
Then $p$ is equivalent to a ballot sequence.
\end{proposition}
Note that the case of zero total weight is excluded here.
In that case we have the following result.
\begin{proposition}\label{prop:july19_1}
Suppose $p \in (B_2)^{\ot L}$ is evolvable and $\wt (p) = 0$.
Then either $p$ or $T_1(p)$ is equivalent to a ballot sequence.
\end{proposition}
For example consider the path
$p=(11)(22)(21)(11)(22)(21)(21)$, that
is not equivalent to a ballot sequence.
For this $p$ we obtain
\begin{displaymath}
T_2T_1(p)=(11)(21)(21)(22)(11)(21)(22),
\end{displaymath}
that is a ballot sequence.
We give a proof of this proposition in Appendix \ref{app:a}
as an application of the main result of this paper.

To find an appropriate formulation of the problem,
next we consider the path
$p=(221)(111)(221)(221)$.
In Example \ref{ex:jun30_2} we derived
\begin{math}
T_1(p)=(211)(211)(211)(222)
\end{math}
and it is easy to see $T_1^2(p)=p$.
Hence there is no $j$ such that $T_1^j(p)$ is equivalent to
a ballot sequence.
However if {$T_2$ actions are allowed}
we {find}
\begin{displaymath}
\batten{(22)}{(211)}{(222)}{}\,
\batten{(11)}{(211)}{(111)}{}\,
\batten{(21)}{(211)}{(211)}{}\,
\batten{(21)}{(222)}{(221)}{(22),}
\end{displaymath}
by {using} the combinatorial $R$ in Example \ref{ex:july30_1},
{hence can obtain} a ballot sequence {as}
$(T_3)^3T_2T_1(p)=(111)(211)(221)(222)$.
Having seen this example one may think that
the following question arises naturally:
(1) Let $s$ be any positive integer and consider $p \in (B_s)^{\ot L}$.
(2) Suppose $p$ is evolvable and of nonnegative total weight.
(3) Then, is it always possible to make such $p$
be equivalent to a ballot sequence by means of
$T_1,T_2,\ldots,T_{s-1}$?

We call it the problem of creating ballot sequences,
and it is
solved affirmatively.
In fact, we have the following theorem.
\begin{theorem}\label{th:may21_1}
Let $p \in (B_s)^{\ot L}$ be an evolvable path satisfying $\wt (p) \geq 0$.
Then there exists $j$ 
such that $T_{s-1}^j (p)$ is equivalent to a ballot sequence.
\end{theorem}
This is our main result.
Note that 
not only the problem is solved affirmatively, but also
it is solved without using $T_r \, (r \leq s-2)$.
For instance if $p=(221)(111)(221)(221)$
then $T_3^2 T_2^3(p)=(111)(211)(221)(222)$,
that again can be checked by using
the combinatorial $R$ in Example \ref{ex:july30_1}

Since $p_{\pm}$ are equivalent to ballot sequences from the outset,
it suffices to prove the theorem for the elements of
$\mathcal{P}_{L,s}$.
Thus we restrict ourselves to considering only these cases
in the subsequent sections.
%
\section{Proof in the Case of Odd Cell Capacity}
\label{sec:3}
\subsection{Ballot sequences and local minima}
\label{subsec:3_1}
{The purpose of this section is to give a proof
of Theorem \ref{th:may21_1} in the case of odd $s$.
To begin, we introduce some notions to formulate the problem
in terms of the weight notation.}
Let $s(\geq 3)$ be an odd integer.
For $b,b' \in B_s$ we call $b \ot b'$ a \textbf{local minimum} 
(resp. a \textbf{local maximum} )
{when} $b<0$ and $b'>0$ (resp. $b>0$ and $b'<0$ ) {are satisfied}.
If $b \ot b'$ is a local minimum or a local maximum, 
{then} it is called an
\textbf{extremum}.
For $p = b_1 \ot \cdots \ot b_L \in (B_s)^{\ot L}$ we define
the functions $w_k, \, W_k:(B_s)^{\ot L}
\rightarrow \Z$ by
\begin{align}
w_k(p) &= \sum_{i=1}^k b_i -\frac12 (s-b_{k+1}),
\label{eq:jun21_2} \\
W_k (p) &= w_k(p) - \frac12 \max (b_k + b_{k+1},0),
\label{eq:jun21_3}
\end{align}
for each integer $k\,(0 \leq k \leq L-1)$.
We understand the index of $b_i$
to be modulo $L$, and write $b_L \ot b_1$ for a
formal juxtaposition of $b_L$ and $b_1$.
We call $p$ a \textbf{highest path} or a \textbf{ballot sequence}
if $w_k(p) \geq 0$ for any $k \,(0 \leq k \leq L-1)$.
In the word notation this definition
of the ballot sequence
amounts to the conventional one in \S \ref{sec:1}.
{}From the obvious recurrence
$w_k(p) = w_{k-1}(p) +\frac12 (b_k + b_{k+1})$
we {obtain} the following results.
\begin{lemma}\label{lem:may29_1}
Suppose one of the following conditions is satisfied.
\begin{itemize}
\item $w_{k-1}(p) \geq 0$ and $b_k + b_{k+1} \geq 0$
\item $w_{k+1}(p) \geq 0$ and $b_{k+1} + b_{k+2} \leq 0$ 
\end{itemize}
Then $w_k(p) \geq 0$.
\end{lemma}
\begin{lemma}\label{lem:july13_5}
The relation $W_k (p) \geq 0$ holds if and only if
$w_k(p) \geq 0$ and $w_{k-1}(p) \geq 0$ hold.
\end{lemma}
\begin{proof}
Denote $b_k + b_{k+1}$ by $B_k$.
Then we have 
$W_k (p) = w_k(p) - \frac12 \max (B_k,0) = 
w_{k-1}(p)  - \frac12 \max (- B_k,0)$.
So if $\Wt (k) \geq 0$ then we have
$w_k(p) \geq \frac12 \max (B_k,0) \geq 0$ and
$w_{k-1}(p)  \geq \frac12 \max (- B_k,0) \geq 0$.
Conversely if $w_k(p) \geq 0$ and $w_{k-1}(p) \geq 0$ hold
then
\begin{math}
W_k (p) \geq \max \left( 
- \frac12 \max (B_k,0), - \frac12 \max (- B_k,0)
\right) = 0.
\end{math}
\end{proof}
For $p = b_1 \ot \cdots \ot b_L \in (B_s)^{\ot L}$ we define
\begin{displaymath}
\mathcal{M}(p) = \left\{
k \in [0,L-1] \,\vert \,\mbox{$b_k \ot b_{k+1}$ is a local minimum.}
\right\}
\end{displaymath}
\begin{proposition}\label{prop:may30_1}
The path $p$ is a ballot sequence if and only if
the following relations hold:
$w_0(p) \geq 0$, $w_{L-1}(p) \geq 0$, and
$W_k (p) \geq 0$ for all $k \in \mathcal{M}(p)$.
\end{proposition}
\begin{proof}
The \textit{only if} direction follows from 
Lemma \ref{lem:july13_5}.
We consider the \textit{if} direction.
Let $b_j \ot b_{j+1}\,(Y)$ and $b_l \ot b_{l+1} \,(Y')\quad (j+1<l)$ be
an adjacent pair of local minima along the path $p$.
By our assumption we have 
$W_j (p), W_l (p) \geq 0$ which implies
$w_j(p) \geq 0$ and $w_{l-1}(p)\geq 0$
by Lemma \ref{lem:july13_5}.
There is a unique local maximum between $Y$ and $Y'$,
say $b_k \ot b_{k+1}$.
We claim if $i$ is between $j+1$ and $l-2$ then
$w_i(p) \geq 0$ holds.
This is proved by induction on $i$:
If $i$ falls between $j+1$ and $k-1$ (resp.~$k$ and $l-2$)
then we use the first (resp.~second) condition
in Lemma \ref{lem:may29_1} and the initial condition
$w_j(p) \geq 0$ (resp.~$w_{l-1}(p)\geq 0$).
So the claim is proved.
There remain the cases in which
either $Y$ is the rightmost local minimum
or $Y'$ is the leftmost one along $p$.
These cases can be treated by
setting $l = L$ or $j=0$ in the above argument.
\end{proof}
As in this proof, we sometimes attach labels to the local minima
and call them by these labels.
When $k \in \mathcal{M}(p)$
we write $\Wt (X)$ for $W_k(p)$ if $X$ is the label for the
local minimum $b_k \ot b_{k+1}$.
We call $\Wt (X)$ the \textbf{modified weight} of $X$.
Let $\mathcal{M}=\{ X_1,\ldots,X_M \}$ be 
the set of all local minima in $p$ denoted by these labels.
For the moment we assume the {following two conditions}.
(1) If $0 \in \mathcal{M}(p)$ then $X_1$ is the label for $b_L \ot b_1$.
(2) {The local minimum} $X_i$ is located to the right of $X_{i-1}$.
{Under this assumption} we define
\begin{equation}\label{eq:jun21_5}
\Wt_{X_i}(X_j) = \Wt (X_j) - \Wt (X_i) +
\begin{cases}
0 & i \leq j, \\
\wt (p) & i > j,
\end{cases}
\end{equation}
which is called the \textbf{modified weight difference}
between $X_i$ and $X_j$.
We call $X \in \mathcal{M}$ a \textbf{virtual global minimum}
if $\Wt_X(Y) \geq 0$ holds for any $Y \in \mathcal{M}$.
In particular, we call such $X$ a \textbf{global minimum}
if $\wt (p) = 0$.
\begin{example}\label{ex:july23_1}
Consider $p = 3 \ot -1 \ot 3 \ot 3 \ot -5 \ot 1 \in (B_5)^{\ot 6}$.
The $p$ has a positive total weight $\wt (p) =4$.
There are two local minima in $p$ at
$\mathcal{M}(p) = \{ 2,5 \}$, which are labeled by
$X_1, X_2$.
Then we obtain $Wt_{X_1}(X_2)=1$ and $Wt_{X_2}(X_1)=3$.
Hence both $X_1$ and $X_2$ are virtual global minima in $p$.
\end{example}
If a local minimum takes the form $-s \ot b$ with some $b \geq 0$,
we call it a \textbf{pre-spike}.
In particular, we call it a \textbf{spike}
if it takes the form $-s \ot s$.
We also say that the local minimum
\textit{takes the spike form} or it 
\textit{is of the spike form}.
\begin{proposition}\label{prop:may25_5}
Let $p \in (B_s)^{\ot L}$ be a path satisfying $\wt (p) \geq 0$.
If $p$ has a (virtual) global minimum taking the spike form,
then it is equivalent to a ballot sequence.
The opposite holds if $\wt (p) = 0$.
\end{proposition}
\begin{proof}
It is enough to show that if 
$p=b_1 \ot \cdots \ot b_L$ has a (virtual) global minimum of the form
$-s \ot s$ at $b_L \ot b_1$ then it is a ballot sequence,
and the opposite holds if $\wt (p) = 0$.
The former part is proved as follows.
Since $w_0(p)= -\frac12 (s - b_1)$ and 
$w_{L-1}(p) = \wt (p) - \frac12 (s + b_L)$, we have
$w_0(p)=0$ and  $w_{L-1}(p) \geq 0$.
Denote the (virtual) global minimum by $X_1$.
Given $X$, any local minimum in $p$,
we have
\begin{equation}\label{eq:july20_2}
\Wt_{X_1}(X) = \Wt(X) - \Wt(X_1).
\end{equation}
But the second term vanishes as
$\Wt(X_1) = W_0 (p) = w_0(p) - \frac12 
\max (b_L + b_1,0) = 0$.
Since $X_1$ is a virtual global minimum, we obtain
$\Wt(X) = \Wt_{X_1}(X) \geq 0$.
Then the claim follows by
Proposition \ref{prop:may30_1}.

The latter part is proved as follows.
Suppose $p$ is a ballot sequence and $\wt (p) = 0$.
Then by the conditions $w_0(p) \geq 0$ and $w_{L-1}(p) \geq 0$
we are forced to have $b_1 = s$ and $b_L = -s$.
It remains to show that the local minimum
$b_L \ot b_1 \, (X_1)$ is a global minimum.
By Proposition \ref{prop:may30_1} the relation
$\Wt(X) \geq 0$ holds for any local minimum $X$.
On the other hand we have $\Wt (X_1) = 0$ 
for the local minimum $X_1$.
Hence by (\ref{eq:july20_2}) we have $\Wt_{X_1}(X) \geq 0$
for any $X$.
\end{proof}
%
\subsection{Conservation of local minima}
\label{subsec:3_2}
Let $p = b_1 \ot \cdots \ot b_L \in (B_s)^{\ot L}$ be a
$T_{s-1}$-evolvable path 
and define $b_k' \, (1 \leq k \leq L)$ by
$T_{s-1}(p) = b'_1 \ot \cdots \ot b'_L $.
Then the process is depicted by the following {transition diagram}
\begin{equation}\label{eq:jun14_2}
\batten{v_0}{b_1}{b_1'}{v_1}\!\!\!
\batten{}{b_2}{b_2'}{v_2}\!\!\!
\batten{}{}{}{\cdots\cdots}
\quad
\batten{}{}{}{v_{L-2}}\,\,
\batten{}{b_{L-1}}{b_{L-1}'}{v_{L-1}}\,\,
\batten{}{b_L}{b_L'}{v_L(=v_0).}
\end{equation}
{Again}, we understand the indices of $b_i$ and $b_i'$
to be modulo $L$, and write $b_L \ot b_1$ and $b'_L \ot b'_1$ for
formal juxtapositions.
\begin{proposition}\label{lem:may16_1}
If $b_k \ot b_{k+1}$ is a local minimum, then either 
$b'_k \ot b'_{k+1}$ or $b'_{k+1} \ot b'_{k+2}$ is a local minimum.
\end{proposition}
\begin{proof}
Suppose $b_k \ot b_{k+1}$ is a local minimum, i.~e.~$b_k <0,\, b_{k+1}>0$.
\begin{equation}\label{eq:july9_1}
\batten{v_{k-1}}{b_k}{b'_k}{v_k} \!\!\! \batten{}{b_{k+1}}{b'_{k+1}}{v_{k+1}}\,
\batten{}{b_{k+2}}{b'_{k+2}}{v_{k+2}}
\end{equation}
First we assume $v_k + b_{k+1} < 0$.
Then by Lemma \ref{lem:july8_1}
we have $b'_{k+1} = v_k - 1 \leq b_k < 0$ and
$b'_{k+2} \geq v_{k+1} - 1 = b_{k+1} > 0$.
Hence $b'_{k+1} \ot b'_{k+2}$ is a local minimum.
Next we assume $v_k + b_{k+1} > 0$.
We claim if $v_k \leq -2$ (resp.~$v_k = 0$, which implies $b_k=-1$) then 
$b'_{k+1} \ot b'_{k+2}$ (resp.~$b'_{k} \ot b'_{k+1}$) is a local minimum.
In the former case we have $b'_{k+1} = v_k + 1 \leq -1$ and
$b'_{k+2} \geq v_{k+1} - 1 = b_{k+1}-2 \geq b_{k+1} + v_k > 0$, while
in the latter $b'_{k+1} = v_k + 1 = 1 > 0$ and
$b'_{k} = v_{k-1} - 1 = v_{k-1} + b_k < 0$.
So the claim is proved, yielding the desired result.
\end{proof}
\begin{coro}\label{cor:july10_3}
The number of local minima is preserved by $T_{s-1}$.
\end{coro}
\begin{proof}
Proposition \ref{lem:may16_1} implies that 
the number of local minima
does not decrease.
Since the time evolution is invertible and $(B_s)^{\ot L}$ is
a finite set, it does not increase also.
\end{proof}

\begin{example}\label{ex:july23_4}
We consider a path $p \in \mathcal{P}_{12,5}$, and show its
time evolution by $T_4$ in Fig.~\ref{fig:1} (weight notation).
In each time there are four 
local minima that are marked by parentheses.
In the word notation the first and the last rows
are written as follows
{\footnotesize
\begin{align*}
p&=
(22111)(11111)(22222)(21111)(22222)(11111)
(22222)(22221)(22111)(22211)(21111)(22111), \\
T_4^5(p) &=
(11111)(21111)(22211)(21111)(22211)(22211)
(22211)(22111)(22211)(22222)(21111)(22222).
\end{align*}
}
We note that the latter is a ballot sequence.
\begin{figure}
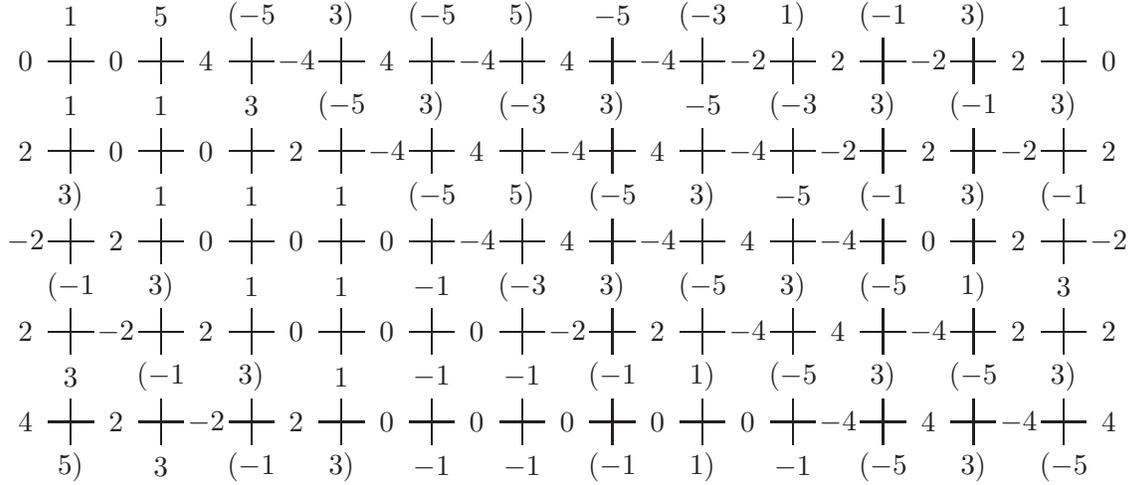

\begin{center}
\latticefig
\end{center}
\caption{A composite transition diagram for
the time evolution of a path 
in $\mathcal{P}_{12,5}$ by $T_4$.}
\label{fig:1}
\end{figure}
\end{example}

\subsection{Conservation of segments}
\label{subsec:3_3}
Given $a \in \Z_{>0}$ we define
\begin{displaymath}
\mathcal{A}(a) = 
\{ (\pm a \ot \mp a )^{\ot n} \,\vert \, n=1,2,\ldots \} \cup
\{ \mp a \ot (\pm a \ot \mp a )^{\ot n} \,\vert \, n=0,1,\ldots \},
\end{displaymath}
where the elements are regarded as the elements
of $(B_s)^{\ot m} \, (m=1,2,\ldots)$ 
written in the weight notation, and let
\begin{math}
\mathcal{A} = \bigcup_{0 \leq a \leq s} \mathcal{A}(a).
\end{math}
%
Given $p \in \mathcal{P}_{L,s}$ 
we can write it uniquely into the form
$p= D_1 \ot \cdots \ot D_{\tilde{L}} \, (\tilde{L} \leq L)$
in which 
the following conditions are satisfied.
\begin{enumerate}
\item Each of $D_i \, (1 \leq i \leq \tilde{L})$ is an
element of $\mathcal{A}$.
\item None of $D_i \ot D_{i+1} \, (1 \leq i \leq \tilde{L}-1)$
is an element of $\mathcal{A}$.
\end{enumerate}
We call {such $D_i$ \textbf{segments}}, except that 
if $D_{\tilde{L}} \ot D_1$ is an element of $\mathcal{A}$
then we regard it as a single segment $D_1$.
If a segment is an element of $\mathcal{A}(a)$ we call it
a level $a$ segment.
A level $s$ segment is called \textbf{maximal}, and
we say $D_i$ \textbf{singular} if its level is less than
$D_{i \pm 1}$'s 
and each of the junctions between $D_i$ 
and $D_{i \pm 1}$ forms an extremum.
\begin{example}\label{ex:jun9_1}
In the first row of Fig.~\ref{fig:1} there are eight segments,
$D_1=1, D_2=5 \ot -5, D_3=3, D_4 = -5 \ot 5 \ot -5,
D_5=-3, D_6=1 \ot -1, D_7=3$, and $D_8=1$.
Among them, $D_2$ and $D_4$ are maximal, while
$D_3$ and $D_6$ are singular segments.
\end{example}

In general the segment content of a path is not preserved by $T_{s-1}$.
However the number of maximal segments and that of
singular segments are preserved.
More precisely we have the following results.
Let the rightmost (resp. leftmost) element of a segment be called 
its \textbf{head} (resp. \textbf{tail}).
Given $p \in \mathcal{P}_{L,s}$ as above let $b_k' \, (1 \leq k \leq L)$ be
the $k$-th element of $T_{s-1}(p)$.
\begin{proposition}\label{lem:may16_2}
If $b_k$ is the head of a maximal segment,
then $b'_{k+1}$ is the tail of a maximal segment.
\end{proposition}
\begin{proof}
Since $b_k$ is the head of a maximal segment, {we have}
$b_k \ot b_{k+1} = s \ot x$ with $x \ne -s$,
or $b_k \ot b_{k+1} = (-s) \ot y$ with $y \ne s$.
First we observe that
for any $a \in B_{s-1}$,
$a + s > 0$ and $a + (-s) < 0$ hold.
We also note that if $x \ne -s$ then $(s-1)+x > 0$, and
if $y \ne s$ then $(-s+1)+y < 0$.
Thus we have the following diagrams by Lemma \ref{lem:july8_1}.
\begin{displaymath}
\batten{a}{s}{a+1}{s-1} \,\,\,
\batten{}{x}{s}{x-1}
\qquad\qquad 
\batten{a}{-s}{a-1}{-s+1} \quad \,\,\,
\batten{}{y}{-s}{y+1}
\end{displaymath}
Since $a+1 \ne -s$ and $a-1 \ne s$ for any $a \in B_{s-1}$,
we obtain the desired result.
\end{proof}

\begin{coro}\label{cor:july10_2}
The number of maximal segments
is preserved by $T_{s-1}$. 
\end{coro}
\begin{proof}
Similar to the proof of Corollary \ref{cor:july10_3}.
\end{proof}
\begin{coro}\label{cor:jun8_2}
If $p  \in (B_s)^{\ot L}$ is
$T_{s-1}$-evolvable in a strong sense, so is $T_{s-1}(p)$.
\end{coro}
\begin{example}\label{ex:july18_3}
Consider the time evolution in Example \ref{ex:july23_4}.
Here we use $\dot{1},\dot{3},\dot{5}$ instead of
$-1,-3,-5$.
\begin{align*}
p &= 1 \ot \underline{5 \ot \dot{5}} \ot 3 \ot 
\underline{\dot{5} \ot 5 \ot \dot{5}} 
\ot \dot{3} \ot 1 \ot \dot{1} \ot 3 \ot 1  \\
T_4(p) &= 1 \ot 1 \ot 3 \ot \underline{\dot{5}} \ot 3 \ot \dot{3} \ot
3 \ot \underline{\dot{5}} \ot \dot{3} \ot 3 \ot \dot{1} \ot 3  \\
T_4^2(p) &= 3 \ot 1 \ot 1 \ot 1 \ot 
\underline{\dot{5} \ot 5 \ot \dot{5}} 
\ot 3 \ot \underline{\dot{5}} \ot \dot{1} \ot 3 \ot \dot{1}  \\
T_4^3(p) &= \dot{1} \ot 3 \ot 1 \ot 1 \ot \dot{1} \ot \dot{3} \ot
3 \ot \underline{\dot{5}} \ot 3 \ot 
\underline{\dot{5}} \ot 1 \ot 3  \\
T_4^4(p) &= 3 \ot \dot{1} \ot 3 \ot 1 \ot \dot{1} \ot \dot{1} \ot
\dot{1} \ot 1 \ot \underline{\dot{5}} \ot 3 \ot 
\underline{\dot{5}} \ot 3  \\
T_4^5(p) &= \underline{5} \ot 3 \ot \dot{1} \ot 3 \ot \dot{1} \ot \dot{1} \ot
\dot{1} \ot 1 \ot \dot{1} \ot 
\underline{\dot{5}} \ot 3 \ot 
\underline{\dot{5} } 
\end{align*}
The maximal segments are marked by underlines.
In $T_4^5(p) $ the combination of rightmost $\dot{5}$ and
leftmost $5$ is understood to be
a single maximal segment $\dot{5} \ot 5$.
\end{example}
We show that the number of singular segments is also preserved.
For this purpose we first give {definitions} and a lemma.
Say $D_i$ {be}
\textbf{left singular} 
{(resp.~\textbf{right singular})}
if its level is less than 
$D_{i-1}$'s 
{(resp.~$D_{i+1}$'s )}
and the junction between $D_{i-1}$ 
{(resp.~$D_{i+1}$)}
and $D_i$ forms an extremum.
Let $p = b_1 \ot \cdots \ot b_L \in (B_s)^{\ot L}$ be 
a $T_{s-1}$-evolvable path and $v_i \in B_{s-1} \, (0 \leq i \leq L)$ 
the elements appeared in the transition diagram (\ref{eq:jun14_2}).
\begin{lemma}\label{lem:jun14_1}
The element $b_k$ belongs to a left singular segment 
if and only if
the element $v_k$ 
is given by $v_k = b_k+1$ for
$b_k > 0$, or $v_k = b_k-1$ for $b_k < 0$.
\end{lemma}
\begin{proof}
Let $D$ be the segment to which $b_k$ belongs, and $a$ the level of $D$.
To prove the \textit{only if} direction
suppose $D$ is left singular.
Then it is enough
to show that the pattern
$\cdots \ot a \ot -a \ot a \ot -a \ot \cdots$
in $D$ becomes
$\cdots \ot -a-2 \ot a+2 \ot -a-2 \ot a+2 \ot \cdots$ by $T_{s-1}$, 
yielding the transition diagram
\begin{equation}\label{eq:jun11_2}
\batten{\cdots}{a}{-a-2}{a+1}\,\,\,
\batten{}{-a}{a+2}{-a-1}\,\,\quad
\batten{}{a}{-a-2}{a+1}\,\,\,
\batten{}{-a}{a+2}{\cdots .}
\end{equation}
This assertion is proved as follows.
Suppose $b_j = \mp a$ is the tail of $D$.
\begin{equation}\label{eq:jun11_3}
\batten{}{b_{j-1}}{}{v_{j-1}} \,
\batten{}{\mp a}{b'_j}{v_j} \!\!\! \batten{}{\pm a}{b'_{j+1}}{\cdots .}\,
\end{equation}
It suffices to show $v_j = \mp (a+1)$.
If $b_j = - a$ then $b_{j-1} \geq a+2$ since $D$ is left singular.
We have $v_{j-1} \geq b_{j-1}-1 \geq a+1$, 
i.~e.~$v_{j-1} + b_j  >0$
which implies $v_j =  -a-1$ by Lemma \ref{lem:july8_1}.
In the same way
if $b_j = a$ then $b_{j-1} \leq -a-2$, yielding $v_j = a+1$.

To prove the \textit{if} direction suppose $D$ is not left singular.
For instance suppose $b_{j-1} \geq -a+2$ and $b_j = a > 0$
in (\ref{eq:jun11_3}).
Then we have $v_{j-1} \geq b_{j-1}-1 \geq -a+1$, 
i.~e.~$v_{j-1} + b_j \geq 1$
which implies $v_j = a-1$.
Similarly if $b_j = -a$ and $b_{j-1} \leq a-2$ then we have
$v_j = -a+1$.
\end{proof}
\begin{remark}
Lemma \ref{lem:jun14_1} becomes irrelevant for $b_k=0$ in the
case of even $s$.
In this case we have Item 4 of Proposition \ref{prop:july18_2}
instead.
\end{remark}
\begin{proposition}\label{lem:jun8_1}
If $b_k$ is the head of a singular segment,
then $b'_{k+1}$ is the tail of a singular segment of the same level.
\end{proposition}
\begin{proof}
Suppose $b_k = a (\ne s)$ is the head of a level $a$ singular segment $D$.
(The case of $b_k = -a$ can be proved in a similar way.)
{Then $b_{k+1} \leq -a-2$ since $D$ is right singular.}

\begin{equation}\label{eq:jun11_1}
\batten{}{b_{k-1}}{}{v_{k-1}} \,
\batten{}{a}{b'_k}{v_k} \!\!\! \batten{}{b_{k+1}}{b'_{k+1}}{v_{k+1}}\,
\batten{}{}{b'_{k+2}}{}
\end{equation}
To begin, we show that the segment to which $b'_{k+1}$ 
belongs is left singular {and level $a$}.
More precisely we claim that $b_k' \leq -a-2$ and $b'_{k+1}=a$.
{(1) Suppose} the length of $D$ is one.
Then $b_{k-1} \leq -a-2$ since $D$ is {left} singular.
{It follows that} $v_{k-1} \leq b_{k-1}+1 \leq -a-1$, 
i.~e.~$v_{k-1} + b_k \leq -1$
which implies $v_k = a+1$ and $b_k' = v_{k-1}-1 \leq -a-2$
by Lemma \ref{lem:july8_1}.
Hence the first claim is confirmed.
{In addition,} we have $v_{k} + b_{k+1} \leq -1$
which implies $b_{k+1}' = v_{k}-1 =a$, hence the second one.
%
{(2) Suppose}
the length of $D$ is more than one.
Then $b_{k-1}=-a$.
By Lemma \ref{lem:jun14_1} we have $b_k'=-a-2$ and $v_k=a+1$,
which implies $b'_{k+1}=a$ as above.
Hence the claim is proved.

Returning to diagram (\ref{eq:jun11_1}),
denote by $D'$ the level $a$ segment in $T_{s-1}(p)$
{to which} $b_{k+1}'(=a)$ {belongs} as its tail.
We are to show that $D'$ is singular.
Since its left singularity {has been} proved,
it remains to show 
{its right singularity.}
{(1) Suppose} $b_{k+1} \leq -a-4$.
Then $b'_{k+2} \leq v_{k+1}+1 \leq b_{k+1}+2 \leq -a-2$.
Hence $b_{k+1}'$ {is} the head of $D'$ and 
{$D'$ is right singular}.
{(2) Suppose} $b_{k+1} = -a-2$.
Since $v_k=a+1$ we have the following diagram.
\begin{equation}\label{eq:jun12_1}
\batten{\cdots}{-a-2}{a}{-a-1}\,\,\quad
\batten{}{a+2}{-a}{a+1}\,\,\,
\batten{}{-a-2}{a}{-a-1}\,\,\quad
\batten{}{a+2}{-a}{\cdots}
\end{equation}
The upper row shows a level $a+2$ segment.
Suppose its head is given by $b_l = a+2$.
(The case of $b_l = -a-2$ can be proved in a similar way.)
It is depicted as follows.
\begin{equation}\label{eq:jun12_2}
\batten{}{a+2}{-a}{a+1} \,\,\, 
\batten{}{b_{l+1}}{b'_{l+1}}{v_{l+1}}\,
\batten{}{b_{l+2}}{b'_{l+2}}{}
\end{equation}
There are two possible cases.
(i) If $b_{l+1} \leq -a-4$ then $(a+1)+b_{l+1} <0$ which implies
$b'_{l+1}=a$.
Also it follows that
$b'_{l+2} \leq v_{l+1}+1 \leq b_{l+1}+2 \leq -a-2$.
Hence $b'_{l+1}(=a)$ {is} the head of $D'$.
(ii) If $b_{l+1} \geq -a$ then $(a+1)+b_{l+1} >0$ which implies
$b'_{l+1}=a+2$.
Hence $b'_{l}(=-a)$ {is} the head of $D'$.
{In both cases $D'$ is right singular.}
The proof is completed.
\end{proof}

\begin{coro}\label{cor:jun9_2}
The number of singular segments of each level
is preserved by $T_{s-1}$. 
\end{coro}
\begin{proof}
Similar to the proof of Corollary \ref{cor:july10_3}.
\end{proof}
\begin{example}\label{ex:july18_4}
Consider the time evolution in Example \ref{ex:july23_4}.
In each time there are two singular segments,
one is level $3$ and the other is level $1$.
\begin{align*}
p &= 1 \ot 5 \ot \dot{5} \ot \ul{3} \ot \dot{5} \ot 5 \ot
\dot{5} \ot \dot{3} \ot \ul{1 \ot \dot{1}} \ot 3 \ot 1  \\
T_4(p) &= 1 \ot 1 \ot 3 \ot \dot{5} \ot 
\ul{3 \ot \dot{3} \ot 3} 
\ot \dot{5} \ot \dot{3} \ot 3 \ot \ul{\dot{1}} \ot 3  \\
T_4^2(p) &= 3 \ot 1 \ot 1 \ot 1 \ot \dot{5} \ot 5 \ot
\dot{5} \ot \ul{3} \ot \dot{5} \ot \dot{1} \ot 3 \ot \ul{\dot{1}}  \\
T_4^3(p) &= \ul{\dot{1}} \ot 3 \ot 1 \ot 1 \ot \dot{1} \ot \dot{3} \ot
3 \ot \dot{5} \ot \ul{3} \ot \dot{5} \ot 1 \ot 3  \\
T_4^4(p) &= 3 \ot \ul{\dot{1}} \ot 3 \ot 1 \ot \dot{1} \ot \dot{1} \ot
\dot{1} \ot 1 \ot \dot{5} \ot \ul{3} \ot \dot{5} \ot 3  \\
T_4^5(p) &= 5 \ot 3 \ot \ul{\dot{1}} \ot 3 \ot \dot{1} \ot \dot{1} \ot
\dot{1} \ot 1 \ot \dot{1} \ot \dot{5} \ot \ul{3} \ot \dot{5} 
\end{align*}
Here the singular segments are marked by underlines.
\end{example}

\subsection{Floating minima and fixed minima}
\label{subsec:3_4}
We show that all of the local minima are classified into
two types according to behaviors of their modified weight values
under the time evolution.
Say a local minimum be \textbf{touched} by a singular segment
{when} it is located inside the singular segment, or at the junction between 
the singular segment and its left or right neighboring segment.
By definition any singular segment touches at least one local minimum
and no local minima are touched by two singular segments.
Then every local minimum is
either a \textbf{floating minimum} or a
\textbf{fixed minimum}.
Say it floating {when} it is the rightmost local
minimum touched by a singular segment.
Otherwise say it fixed.
{}From Corollaries  \ref{cor:july10_3} and
\ref{cor:jun9_2} we have the following.
\begin{proposition}\label{prop:july10_4}
For any $p \in \mathcal{P}_{L,s}$
the number of floating minima 
and that of fixed minima in $p$ are preserved by $T_{s-1}$. 
\end{proposition}

\begin{example}\label{ex:july18_1}
Consider the time evolution in Example \ref{ex:july23_4}.
\begin{align*}
p &= 1 \ot 5 \ot \dot{5} \otb {3} \ot \dot{5} \otv 5 \ot
\dot{5} \ot \dot{3} \otn {1 \ot \dot{1}} \otb 3 \ot 1 \, (\ot) \\
T_4(p) &= 1 \ot 1 \ot 3 \ot \dot{5} \otv 
{3 \ot \dot{3} \otb 3} 
\ot \dot{5} \ot \dot{3} \otn 3 \ot {\dot{1}} \otb 3 \, (\ot) \\
T_4^2(p) &= 3 \ot 1 \ot 1 \ot 1 \ot \dot{5} \otv 5 \ot
\dot{5} \otb {3} \ot \dot{5} \ot \dot{1} \otn 3 \ot {\dot{1}} \, (\otb) \\
T_4^3(p) &= \dot{1} \otb 3 \ot 1 \ot 1 \ot \dot{1} \ot \dot{3} \otv
3 \ot \dot{5} \otb {3} \ot \dot{5} \otn 1 \ot 3 \, (\ot) \\
T_4^4(p) &= 3 \ot {\dot{1}} \otb 3 \ot 1 \ot \dot{1} \ot \dot{1} \ot
\dot{1} \otv 1 \ot \dot{5} \otb {3} \ot \dot{5} \otn 3 \, (\ot) \\
T_4^5(p) &= 5 \ot 3 \ot {\dot{1}} \otb 3 \ot \dot{1} \ot \dot{1} \ot
\dot{1} \otv 1 \ot \dot{1} \ot \dot{5} \otb {3} \ot \dot{5} \, (\otn)
\end{align*}
Here $\bullet$ denotes a floating minimum,
$\vee$ or $\nabla$ a fixed minimum.
In particular $\nabla$ denotes a \textit{lowest} fixed minimum
that will be defined in subsec.~\ref{subsec:3_6}.
\end{example}
\begin{remark}
These names (i.~e.~\textit{floating} 
or \textit{fixed}) are based on
a property of the modified weight difference.
We show that the modified 
weight difference (\ref{eq:jun21_5}) between 
two fixed minima 
does not change by $T_{s-1}$ (Proposition \ref{prop:may23_1}).
In contrast the modified weight difference 
between a fixed minimum and a floating one,
as well as between two floating minima in general, must change.
As was seen in the above example
their positions along the path are not fixed
for both types of local minima.
\end{remark}

Let $p \in \mathcal{P}_{L,s}$ be a path having
$M$ local minima in total.
Suppose the local minima are attached by labels,
$X_1,\ldots,X_M$ say.
Due to Corollary \ref{cor:july10_3} we can consider
one to one {correspondence} between the set of local minima
in $T_{s-1}(p)$ and the set of labels $\{ X_1,\ldots,X_M \}$.
This enables us to regard the local minima 
as distinguishable particles moving along a path.
We introduce two different types of such {correspondence}
which we call the first and the second labeling schemes.
\begin{description}
\item[First scheme]
For any $i$ we
attach label $X_i$ to the local minimum 
located at the same or the right 
neighboring site of the position of $X_i$ in $p$.
\item[Second scheme]
To begin with we apply the first scheme.
Then repeat the following procedure until all singular 
segments in $T_{s-1}(p)$ are treated.
(1) Pick one of the singular segments.
(2) Apply a left cyclic shift to the labels  
of the local minima touched by it.
\end{description}
\begin{example}
We consider the $p$ and $T_4(p)$ in Example \ref{ex:july18_1}.
In the latter their is a level $3$ singular segment 
touching two local minima.
So we exchange their names in the second labeling scheme.
\begin{align*}
p &= 1 \ot 5 \ot \dot{5} \stackrel{X_1}{\ot} 3 \ot \dot{5} \stackrel{X_2}{\ot} 5 \ot
\dot{5} \ot \dot{3} \stackrel{X_3}{\ot} 1 \ot 
\dot{1} \stackrel{X_4}{\ot} 3 \ot 1 \\
T_4(p) &= 1 \ot 1 \ot 3 \ot \dot{5} \stackrel{X_1}{\ot} 3 \ot 
\dot{3} \stackrel{X_2}{\ot}
3 \ot \dot{5} \ot \dot{3} \stackrel{X_3}{\ot} 3 \ot 
\dot{1} \stackrel{X_4}{\ot} 3 \quad \mbox{(first scheme)}\\
T_4(p) &= 1 \ot 1 \ot 3 \ot \dot{5} \stackrel{X_2}{\ot} 3 \ot 
\dot{3} \stackrel{X_1}{\ot}
3 \ot \dot{5} \ot \dot{3} \stackrel{X_3}{\ot} 3 \ot 
\dot{1} \stackrel{X_4}{\ot} 3 \quad \mbox{(second scheme)}
\end{align*}
\end{example}
Let $X_i$ be a floating minimum in $p$.
Then in $T_{s-1}(p)$ the local minimum 
called $X_i$ in the first labeling scheme
is not always floating.
According to Proposition \ref{lem:jun8_1}
it is rather the \textit{leftmost} local minimum touched by a
singular segment.
The cyclic shift in the second scheme changes
it into the \textit{rightmost} local minimum.
Thus the types (floating or fixed) of a local minimum is
preserved in the second scheme.

For the moment we adopt the first labeling scheme, and classify
the moves of the local minima into three types.
Given $a \in \{ 3,5,\ldots , s\}$ we fix $b \in \{ 1,3,\ldots ,a-2\}$.
By the time evolution $T_{s-1}$ we have
\begin{align}
& 
\batten{\mbox{\textbf{I.}}\qquad}{-x}{}{\,-b-1} \quad \,\,\, \batten{}{a}{-b}{\,a-1} \quad \! \batten{}{}{y}{}
\label{eq:apr11_1}
\\
&
\batten{\mbox{\textbf{II.}} \qquad}{-y}{}{\,-a+1} \quad \,\,\, \batten{}{b}{-a}{\,b+1} \quad \! \batten{}{}{x}{}
\label{eq:apr11_2}
\end{align}
where $x=b$ or $b+2$ and $y=a$ or $a-2$.
We write $(|b|,b')$ for the local minimum $b \ot b'$.
Let the move $(x,a) \mapsto (b,y)$ in (\ref{eq:apr11_1}) 
be called \textbf{type I} and
$(y,b) \mapsto (a,x)$ in (\ref{eq:apr11_2}) be \textbf{type II}.
There is another type of move such that
\begin{equation}\label{eq:apr11_3}
\batten{\mbox{\textbf{III.}} \qquad -c'+1}{-1}{-c'}{\,\, 0} 
\!\!\! \batten{}{c}{1}{c-1}
\end{equation}
where $c,c' \in \{ 1,3,\ldots,s\}$. 
Let $(1,c) \mapsto (c',1)$ in (\ref{eq:apr11_3}) be called \textbf{type III}.

\begin{remark}\label{rem:july24_1}
Let $b_k \ot b_{k+1}$ be a local minimum and recall (\ref{eq:july9_1}).
In the above classification the value of $v_k$ is also taken into
account.
We call the triple $( b_k,b_{k+1},v_k )$ a \textit{configuration}
of the local minimum.
If $v_k = 0$, then $b_k = -1$ and 
$b_{k+1} \in \{ 1,3,\ldots, s \}$.
So there are $\frac{s+1}{2}$ such configurations
that are classified into type III move.
If $v_k \leq -2$, then $b_k = v_k \pm 1$,
$v_k \in \{-2,-4,\ldots , -s+1 \}$, and
$b_{k+1} \in \{ 1,3,\ldots, s \}$.
So there are $\frac{s^2 - 1}{2}$ such configurations.
Half of them with the condition $v_k + b_{k+1}>0$
are classified into type I move, and the others
are into type II move.
Here we observed that there are $\frac{s(s+1)}{2}$ configurations
in total, while
in the case of even $s$ we will find there are 
$\frac{s(s+2)}{2}$ configurations (Remark \ref{rem:july24_2}).
\end{remark}
The following result,
which will be used in the next subsection 
(Proposition \ref{prop:may23_1}),
shows a relation between the types of
local minima and the types of their moves.
\begin{lemma}\label{lem:may16_3}
Let $b_k \ot b_{k+1}$ be a local minimum.
\begin{enumerate}
\item Suppose $v_k + b_{k+1} <0$.
Then $b_k \ot b_{k+1}$ is a floating (resp.~fixed) minimum
if and only if
$b'_{k+2} = v_{k+1} - 1$ (resp.~$b'_{k+2} = v_{k+1} + 1$).
\item Suppose $v_k + b_{k+1} >0$.
Then $b_k \ot b_{k+1}$ is a floating (resp.~fixed) minimum
if and only if
$b_{k} = v_{k} + 1$ (resp.~$b_{k} = v_{k} - 1$).
\end{enumerate}
\end{lemma}
\begin{proof}
In view of Remark \ref{rem:july24_1} we find
it is enough to show that a
local minimum is floating if and only if it takes 
type I or II move
with $x=b$ in (\ref{eq:apr11_1}),  (\ref{eq:apr11_2}).
First we consider the case of type I move.
Let $b_k \ot b_{k+1} = -x \ot a$.
(1) Suppose $x=b+2$.
By Lemma \ref{lem:jun14_1} we see that
neither $b_k$ nor $b_{k+1}$ belongs to a (left) singular segment.
Thus $b_k \ot b_{k+1}$ is touched by 
no singular {segments}, hence fixed.
(2) Suppose $x=b$.
By Lemma \ref{lem:jun14_1} and since $x < a$
we find that
$b_k = -b$ is the head of a singular segment.
Thus $b_k \ot b_{k+1}$ is the rightmost local minimum
touched by a singular segment, 
hence floating.

Next we consider the case of type II move.
\begin{equation}\label{eq:jun13_1}
\batten{}{-y}{}{\,-a+1} \quad \,\,\, \batten{}{b}{-a}{\,b+1} 
\quad \! \batten{}{z}{x}{}
\end{equation}
Let $b_k \ot b_{k+1} = -y \ot b$, and
denote by $D$ the segment to which $b_{k+1}(=b)$ belongs.
Then since $y \geq b$ the element {$b_k (= -y)$} cannot be the head of a
singular segment, while by Lemma \ref{lem:jun14_1} the
segment $D$ is left singular.
(1) Suppose $x=b+2$.
It implies $(b+1)+z >0$ in (\ref{eq:jun13_1}), 
i.~e.~$z \geq -b$.
If $z > -b$ 
then either $b \ot z$ is not an extremum or
the level of the segment to which $z$ belongs
is $< b$.
In any case $D$ is not singular.
On the other hand, if $z = -b$ then $D$ could be singular.
But even if so, $b_k \ot b_{k+1}$ cannot be the rightmost local minimum
touched by $D$.
Thus in both cases $b_k \ot b_{k+1}$ is a fixed minimum.
(2) Suppose $x=b$.
By Lemma \ref{lem:july8_1}
it implies $(b+1)+z <0$ in (\ref{eq:jun13_1}), 
i.~e.~$z \leq -b-2$.
Thus $D$ is singular,
and $b_k \ot b_{k+1}$ is the rightmost local minimum
touched by $D$.
Hence it is a floating minimum.

Finally we consider the case of type III move.
Let $b_k \ot b_{k+1} = -1 \ot c$.
By Lemma \ref{lem:jun14_1} we see that
neither $b_k$ nor $b_{k+1}$ belongs to a (left) singular segment.
Thus $b_k \ot b_{k+1}$ is touched by 
no singular segments, hence fixed.
\end{proof}
\subsection{Conservation of modified weight differences}
\label{subsec:3_5}
In Proposition \ref{prop:may30_1} we characterized the ballot sequences
by means of the modified weight.
The purpose of this subsection is to show that the
modified weight difference between two fixed minima is
preserved.
That will be completed in Propositions \ref{prop:may23_1} 
and \ref{rem:jul1_1}.
To establish them we first derive a result 
on the ``not modified" weight difference
in Lemma \ref{prop:apr12_1}.

Recall the diagram (\ref{eq:jun14_2}).
{Again} we adopt the first labeling scheme.
Let $b_j \ot b_{j+1} \, (X)$ and $b_k \ot b_{k+1} \, (Y)$ 
be two local minima which move to 
$b'_l \ot b'_{l+1} \, (X)$ and 
$b'_m \ot b'_{m+1} \, (Y)$,
where $l = j, j+1 $ and $m = k, k+1$ are assumed.
Under this setting we define the functions
$\wt_X (Y), \wt'_{X} (Y)$,
$\Wt_X (Y), \Wt'_{X} (Y)$ by
\begin{align}
\wt_X (Y) &= b_{j+1} + b_{j+2} + \cdots + b_{k},
\\
\wt'_{X} (Y) &= b'_{l+1} + b'_{l+2} + \cdots + b'_{m},\\
\Wt_X(Y) &= \wt_X (Y) + 
\varphi (b_k, b_{k+1},b_j, b_{j+1}),\label{eq:jun21_4}\\
\Wt'_{X}(Y) &= \wt'_{X} (Y) + 
\varphi (b_m', b_{m+1}',b_l', b_{l+1}'),
\end{align}
where
$\varphi (a,b,c,d) = \frac12 ( \min(-a,b) - \min(-c,d) )$.
Also we set $\Delta \wt_X (Y) = \wt'_{X} (Y) - \wt_X (Y) $
and $\Delta \Wt_X (Y) = \Wt'_{X} (Y) - \Wt_X (Y) $.

\begin{remark}
Since we are assuming the periodic boundary condition,
the index $i$ of $b_i$ is interpreted by modulo $L$.
Then (\ref{eq:jun21_4})
is equivalent to the modified weight difference (\ref{eq:jun21_5}).
\end{remark}

\begin{lemma}
\label{prop:apr12_1}
Under the above setting $\Delta \wt_X (Y)$ is given by Table \ref{tab:2}.

\begin{table}[h]
\begin{center}
	\begin{tabular}[h]{|c|c|c|c|}
		\hline
		& I & II & III\\
		\hline
		I & $0$ & $-2$ &  $-1$ \\
		\hline
		II & $+2$ & $0$ & $+1$ \\
		\hline
		III & $+1$ & $-1$ & $0$ \\
		\hline
	\end{tabular}
\end{center}
	\caption{$\Delta \wt_X (Y)$: The first column shows the types of move 
	for $X$ and the top row shows those for $Y$.}
	\label{tab:2}
\end{table}
\end{lemma}

\begin{proof}
If both $X$ and $Y$ take type I or II move we have
\begin{equation}\label{eq:may25_1}
\batten{}{b_j}{}{\,v_j} \!\!\! \batten{}{b_{j+1}}{b'_{j+1}}{\,v_{j+1}}  \,\,
\batten{}{}{b'_{j+2}}{\cdots} \!\!
\batten{}{b_{k}}{}{\,v_k} \!\!\! \batten{}{b_{k+1}}{b'_{k+1}}{\,v_{k+1}} \,\,
\batten{}{}{b'_{k+2}}{} 
\end{equation}
where $v_j = b'_{j+1} \mp 1$ and $v_{k+1} = b_{k+1} \mp 1$.
Throughout this proof the upper (resp.~lower) sign is 
for type I (resp.~type II) move.
Since $v_j + \wt_X(Y)  + b_{k+1} = b'_{j+1} + \wt'_X(Y) + v_{k+1}$
we have $\Delta \wt_X(Y) = v_j -b'_{j+1} - (v_{k+1}-b_{k+1})$
which yields Table \ref{tab:2} except the last column and the last row.
If $X$ takes type III move and $Y$ takes type I or II move we have
\begin{equation}\label{eq:may25_2}
\batten{}{b_j}{b'_j}{\,0} \!\!\! \batten{}{b_{j+1}}{b'_{j+1}}{\,v_{j+1}}  \,\, 
\batten{}{}{}{\cdots} \!\!
\batten{}{b_{k}}{}{\,v_k} \!\!\! \batten{}{b_{k+1}}{b'_{k+1}}{\,v_{k+1}} \,\,
\batten{}{}{b'_{k+2}}{} 
\end{equation}
where $v_{k+1} = b_{k+1} \mp 1$.
Since $\wt_X(Y) + b_{k+1} = \wt'_X(Y) + v_{k+1}$
we have $\Delta \wt_X(Y) = b_{k+1}-v_{k+1} = \pm 1$.
If $X$ takes type I or II move and $Y$ takes type III move we have
\begin{equation}\label{eq:july13_6}
\batten{}{b_j}{}{\,v_j} \!\!\! \batten{}{b_{j+1}}{b'_{j+1}}{\,v_{j+1}}  \,\, 
\batten{}{}{b'_{j+2}}{\cdots} \!\!
\batten{}{}{}{\,v_{k-1}} \,\, \batten{}{b_{k}}{b'_k}{\,0} \!\!\! 
\batten{}{b_{k+1}}{b'_{k+1}}{v_{k+1}} 
\end{equation}
where $v_j = b'_{j+1} \mp 1$.
Since $v_j + \wt_X(Y) = b'_{j+1} + \wt'_X(Y)$
we have $\Delta \wt_X(Y) = v_j -b'_{j+1} = \mp 1$.
Finally if both $X$ and $Y$ take type III move we have the following diagram
which yields $\Delta \wt_X(Y) = 0$.
\begin{displaymath}
\batten{}{b_j}{b'_{j}}{\,0} \!\!\! \batten{}{b_{j+1}}{b'_{j+1}}{\,v_{j+1}}  \,\, 
\batten{}{}{}{\cdots} \!\!
\batten{}{}{}{\,v_{k-1}} \,\, \batten{}{b_{k}}{b'_k}{\,0} \!\!\! 
\batten{}{b_{k+1}}{b'_{k+1}}{} 
\end{displaymath}
\end{proof}
{}From Lemmas \ref{lem:may16_3} and \ref{prop:apr12_1}
we obtain the following result.
\begin{proposition}\label{prop:may23_1}
Let $X,Y$ be a pair of local minima in $p$.
If both $X$ and $Y$ are fixed minima, 
then $\Delta \Wt_X(Y) = 0$
(in the first labeling scheme).
\end{proposition}
\begin{proof}
It is easy to see that
\begin{displaymath}
\Delta \Wt_X(Y) = \Delta \wt_X(Y) + \varphi (b_m',b_{m+1}',b_k,b_{k+1})
-\varphi (b_l',b_{l+1}',b_j,b_{j+1}).
\end{displaymath}
Here the second term is associated with the move of $Y$ and the
third one with $X$.
Suppose $Y$ is a fixed minimum.
Recall that
in (\ref{eq:apr11_1}) and (\ref{eq:apr11_2}) 
it was assumed that $a \geq b+2$ and $y \geq b$.
Then by Lemma \ref{lem:may16_3}
$\varphi (b_m',b_{m+1}',b_k,b_{k+1})$ takes the value $-1$ if the move
is of type I (with $x=b+2$ in (\ref{eq:apr11_1})),
$+1$ if type II (with $x=b+2$ in (\ref{eq:apr11_2})), and $0$ if type III.
Clearly the same result holds for
$\varphi (b_l',b_{l+1}',b_j,b_{j+1})$ if $X$ is a fixed minimum.
Combining these results we obtain a table for the values of
$\varphi (b_m',b_{m+1}',b_k,b_{k+1})-\varphi (b_l',b_{l+1}',b_j,b_{j+1})$
which would agree with Table \ref{tab:2} (Lemma \ref{prop:apr12_1})
if all signs were changed.
\end{proof}
\begin{proposition}\label{rem:jul1_1}
The assertion in
Proposition \ref{prop:may23_1}
holds also in the second labeling scheme.
\end{proposition}
This fact is confirmed by the following result.
\begin{lemma}\label{lem:jun13_1}
The relation
$\wt_X(Y) = \Wt_X(Y)=0$ holds for every pair of local minima $X,Y$
touched by a common singular segment.
\end{lemma}
\begin{proof}
If $X$ and $Y$ are touched by a common segment, then
we have $\wt_X(Y)=0$.
Suppose $b_j \ot b_{j+1} \, (X)$ and $b_k \ot b_{k+1} \, (Y)$ are
touched by a common \textit{singular} segment of level $a$.
Then $\min (-b_j,b_{j+1}) = \min (-b_k,b_{k+1}) = a$,
hence {$\Wt_X(Y) = \wt_X(Y)=0$}.
\end{proof}
\subsection{Spike formation at fixed minima}
\label{subsec:3_6}
The purpose of this subsection is to show that
any fixed minimum takes the spike form 
after applying $T_{s-1}$ appropriate times.
We begin by studying the behavior of the weight difference
between a local minimum and the head of a maximal segment.
In particular we consider the case where 
the latter travels exactly one site by $T_{s-1}$.
Given a path $p = b_1 \ot \cdots \ot b_L \in \mathcal{P}_{L,s}$ 
we suppose $b_j \ot b_{j+1} \, (X)$ is a local minimum.
Let $D$ be a common label for the
maximal segments in $p$ and $T_{s-1}(p)$
such that $D$'s tail in $T_{s-1}(p)$
is located at the right neighboring site of $D$'s head in $p$.
(See Proposition \ref{lem:may16_2}.)
Let the heads of the maximal segments be $b_k$ and $b'_m$ 
in $p$ and $T_{s-1}(p)$ respectively,
and suppose $X$ moves to
$b'_l \ot b'_{l+1}$ in $T_{s-1}(p)$
in the first labeling scheme.
Under this setting we define the functions $\wt_{X} (D), \wt'_{X} (D)$ by
\begin{align}
&
\wt_{X} (D) = b_{j+1} + b_{j+2} + \cdots + b_{k},
\\
&
\wt'_{X} (D) = b'_{l+1} + b'_{l+2} + \cdots + b'_{m}.
\end{align}
Also we set $\Delta \wt_{X} (D) = \wt'_{X} (D) - \wt_{X} (D)$.

Suppose $D$ travels exactly one site by $T_{s-1}$.
If its head is $\pm s$ then the possible types of moves are
given by
\begin{align}
&
\batten{}{\pm s}{}{\pm (s-1)} \quad \quad \,\,\, \batten{}{\mp x}{\pm s}{\mp (x+1)} \quad \quad \,\,\,  \batten{}{}{\mp y}{}
\label{eq:may25_3}
\\
&
\batten{}{\pm s}{}{\pm (s-1)} \quad \quad \,\,\, \batten{}{\mp (s-2)}{\pm s}{\mp (s-1)} \quad \quad \,\,\, \batten{}{\pm s}{\mp (s-2)}{}
\label{eq:may25_4}
\end{align}
where $0 \leq x \leq s-4$ and $0 \leq y \leq s-2$.
By this observation we obtain the following result.
\begin{lemma}
\label{prop:may23_2}
Suppose maximal segment $D$ travels exactly one site by $T_{s-1}$.
Then in both labeling schemes
$\Delta \wt_{X} (D)$ is given by Table \ref{tab:3}
for any local minimum $X$.

\begin{table}[h]
\begin{center}
	\begin{tabular}[h]{|c|c|c|}
		\hline
		& $s$ & $-s$ \\
		\hline
		I & $0$ & $-2$  \\
		\hline
		II & $+2$ & $0$  \\
		\hline
		III & $+1$ & $-1$  \\
		\hline
	\end{tabular}
\end{center}
	\caption{$\Delta \wt_{X} (D) $: The first column shows the types of moves 
	for $X$ and the top row shows $D$'s head.}
	\label{tab:3}
\end{table}
\end{lemma}
\begin{proof}
To begin with we adopt the first labeling scheme.
Then the proof proceeds by comparison with 
that of Lemma \ref{prop:apr12_1},
in which we replace $Y$ by $D$.
If the types of move for
$X$ is I or II then the situation is given by
(\ref{eq:may25_1}), while if it is III then by
(\ref{eq:may25_2}).
If the move of the head of $D$ is given by
(\ref{eq:may25_3}) then we have $b_{k+1}=\mp x, v_{k+1} = \mp
(x+1)$, while if it is given by 
(\ref{eq:may25_4}) then $b_{k+1}=\mp (s-2), 
v_{k+1} = \mp (s-1)$.
This yields $v_{k+1} = b_{k+1} \mp 1$ 
for $D$ with its head $\pm s$.
Compared with the proof of
Lemma \ref{prop:apr12_1}, we find that
if $D$'s head is $s$ (resp.~$-s$) then it corresponds to 
type I (resp.~type II)
move of $Y$ in that lemma.
The assertion holds also in the second
labeling scheme due to the fact shown in Lemma \ref{lem:jun13_1}.
\end{proof}
{}From now on we adopt the second labeling scheme.
We derive three lemmas to obtain a proposition
on the spike formation at fixed minima.
\begin{lemma}\label{prop:jun22_3}
Let $p \in \mathcal{P}_{L,s}$ be a path 
{with} at least one fixed minimum.
Then by applying $T_{s-1}$ repeatedly
a maximal segment catches up with any
fixed minimum in $p$.
\end{lemma}
\begin{proof}
First we consider the case $\wt (p) = 0$.
Let $D$ be a maximal segment, and $X$ a fixed minimum in $p$.
Here we assume $D$'s head is $s$.
(The other case where $D$'s head is $-s$ can be proved similarly.)
By applying $T_{s-1}$,
the (head of the) maximal segment travels at least one {site}, while the
fixed minimum does at most one.
Suppose $D$ and $X$ continue to
propagate by one site per one unit time persistently.
Then $D$'s head remains $s$ and no floating minima overtake $X$.
{}From Lemma \ref{prop:may23_2}
we find that after finite time steps
$X$ continues to take type I moves only, 
or $\wt_{X}(D)$ becomes
arbitrarily large.
But this is impossible because of the following reasons.
(1) The type I move at a fixed minimum is
given by $(b+2,a) \mapsto (b,a)$ or $(b+2,a) \mapsto (b,a-2)$
with $0 \leq b < a \leq s$, 
so it cannot continue persistently.
(2) The weight difference $\wt_{X}(D)$ cannot become
arbitrarily large in a periodic path of zero total weight.
Thus $D$ and $X$ cannot continue to propagate with the same speed
persistently, hence $D$ must catch up with $X$.

Next we consider the case $\wt (p) \ne 0$.
The only difference is that $|\wt_{X}(D)|$ can become
arbitrarily large now.
But if so, the distance between $X$ and the head of $D$ must
exceeds $L$, forcing $D$ 
to catch up with $X$ because of the periodic boundary condition.
\end{proof}
By the time evolution $T_{s-1}$
a fixed minimum moves by $1$, $0$, or $-1$
(See Example \ref{ex:july18_1}),
{and}
a maximal segment moves forwards without
``leaping" over a finite gap
(Proposition \ref{lem:may16_2}).
This implies that
whenever a maximal segment overtakes a fixed minimum, 
they must overlap each other.
More precisely we have the following.
\begin{lemma}\label{prop:july23_2}
A maximal segment does not overtake a fixed minimum
unless they make a (pre-)spike.
\end{lemma}
\begin{proof}
In view of the above observation we find
it suffices to check the following case.
Recall the transition diagram (\ref{eq:jun14_2}).
Suppose $b_{k-1}$ is the head of a maximal segment $D$
and $b_k \ot b_{k+1}$ is a fixed minimum $X$.
Then in the next time step the tail of $D$ is at $b_k'$
(Proposition \ref{lem:may16_2}), hence $|b'_{k}| = s$.
Therefore if
$X$ goes backwards, i.~e.~if it moves to
$b'_{k-1} \ot b'_{k}$ then it may not be a pre-spike.
But this cannot occur because if
$X$ goes backwards 
then $b'_{k}$ is within a singular segment, hence $|b'_{k}| \ne s$.
\end{proof}

\begin{lemma}\label{lem:july24_3}
If a fixed minimum is a pre-spike,
then it changes into a spike by applying $T_{s-1}$
repeatedly.
\end{lemma}
\begin{proof}
If the fixed minimum is already a spike, there is none to be done.
Suppose it has the form $-s \ot b \, (b \ne s)$.
Then by Lemma \ref{lem:july8_1} and Item 1 of
Lemma \ref{lem:may16_3} the time evolution
$T_{s-1}$ yields the move $(s,b) \mapsto (s,b+2)$.
By repeating this we obtain
$(s,b) \mapsto (s,b+2) \mapsto \cdots \mapsto (s,s)$ at 
this fixed minimum.
\end{proof}

As a result of these lemmas we obtain the following result.
\begin{proposition}\label{prop:july27_1}
Let $p \in \mathcal{P}_{L,s}$ be a path {with} at least one fixed minimum.
Then any fixed minimum in $p$ takes the spike form 
after applying $T_{s-1}$ appropriate times.
\end{proposition}
\subsection{The case of zero total weight}
\label{subsec:3_7}
In this subsection we establish 
Theorem \ref{th:may21_1} in the case of $\wt(p)=0$. 
Let $\mathcal{M}=\{ X_1,\ldots,X_M \}$ be 
the set of all local minima in $p$.
Note that in this case we have 
$\Wt_{X_i}(X_j) = - \Wt_{X_j}(X_i)$ for any $i,j$.
The arguments in
this (as well as the following) subsection
are valid in both odd $s$ and even $s$ cases.
\begin{lemma}\label{lem:jun21_1}
There is at least one fixed minimum in a path of zero total weight.
\end{lemma}
\begin{proof}
Suppose there exists a path $p$ of 
zero total weight but has no fixed minima.
Since $p$ is of zero total weight it must have one or more
local minima which are inevitably floating by our assumption.
It means that $p$ has at least one singular segment.
Then each of the singular segments in $p$ must touch exactly one local minimum,
since otherwise a second local minimum touched by it must be fixed.
The types of {such} singular segments are restricted to the
following ones.
\begin{align}
\mbox{\textbf{A.}} & \qquad b \ot (-a \ot b')   \nonumber \\
\mbox{\textbf{B.}} & \qquad b \ot (-a \ot a) \ot -b'  \nonumber \\
\mbox{\textbf{C.}} & \qquad  (-b \ot a) \ot -b' 
\label{eq:may18_1}
\end{align}
Here $b,b' > a \geq 0$.
Pick all the local minima in $p$ (as a non-periodic path) 
from left to right and write their types in a sequence.
Clearly there cannot be such juxtapositions as
\textbf{B}\textbf{A}, 
\textbf{C}\textbf{B}, 
\textbf{C}\textbf{A}, or
\textbf{B}\textbf{B} in the sequence.
Thus the possible cases are as follows.
\begin{align}
\mbox{\textbf{1.}} & \qquad  \underbrace{\mbox{\textbf{A}} \ldots . \mbox{\textbf{A}}}_{n}
\,\underbrace{\mbox{\textbf{C}} \ldots . \mbox{\textbf{C}}}_{m} \nonumber \\
\mbox{\textbf{2.}} & \qquad  \underbrace{\mbox{\textbf{A}} \ldots \mbox{\textbf{A}}}_{n}
\mbox{\textbf{B}}\underbrace{\mbox{\textbf{C}} \ldots \mbox{\textbf{C}}}_{n}
\label{eq:jun20_1}
\end{align}
Since $p$ is of zero total weight it cannot take
the form \textbf{1} with $n=0$ or $m=0$.
Then considering $p$ as a periodic path we are forced to have 
one of
\textbf{B}\textbf{A}, 
\textbf{C}\textbf{B}, or
\textbf{C}\textbf{A},
hence a contradiction.
\end{proof}
Let $\mathcal{M}'$ be a subset of $\mathcal{M}$.
If $X \in \mathcal{M}'$ satisfies $\Wt_X(Y) \geq 0$ for any
$Y \in \mathcal{M}'$, 
we call $X$ a \textbf{lowest minimum}
associated with $\mathcal{M}'$.
Thus a lowest minimum associated with $\mathcal{M}$
is a global minimum.
We show that it is always possible to find
a global minimum that is also a fixed minimum.
Let $\mathcal{M}_{\rm FX}$ be the set of all fixed minima
and $\mathcal{M}_{\rm FL}=\mathcal{M} \setminus \mathcal{M}_{\rm FX}$, 
the set of all floating minima in $p$.
A lowest minimum associated with $\mathcal{M}_{\rm FX}$
is called a \textbf{lowest fixed minimum}.
\begin{lemma}\label{prop:may18_3}
A lowest fixed minimum is a global minimum.
\end{lemma}
\begin{proof}
Let $\mathcal{M}_{\rm FX} = \{ Y_1,\ldots,Y_N \}$ and
$\mathcal{M}_{\rm FL} =\{  Z_1,\ldots,Z_{N'} \}$.
By reshuffling their names we can assume that
$Y_1$ and $Z_1$ are lowest minima associated with 
$\mathcal{M}_{\rm FX}$ and $\mathcal{M}_{\rm FL}$ respectively.
Then it suffices to show $\Wt_{Y_1} (Z_1) \geq 0$.
The floating minimum
$Z_{1}$ is touched by a singular segment, which we call $D$.
If there is another local minimum touched by $D$, then it
is a fixed minimum, say $Y_2$.
It follows that $\Wt_{Y_1} (Z_{1}) = \Wt_{Y_1} (Y_2) 
+ \Wt_{Y_2} (Z_1) \geq 0$ by Lemma \ref{lem:jun13_1},
so we are done.
Suppose there are no local minima touched by $D$ other 
than $Z_{1}$ itself.
Then $Z_{1}$ is in one of the three cases in (\ref{eq:may18_1}).
For instance consider type \textbf{A} case.
(The other cases can be proved in a similar way.)
Depicting the sequence of its left side elements explicitly, it 
generally looks like
\begin{displaymath}
(-c_0 \ot c_1) \ot \cdots \ot c_l \ot (-a \ot b'),
\end{displaymath}
where $c_i \geq 0 \, (0 \leq i \leq l-1)$, $c_l=b$,
and $-c_0 \ot c_1$ is the adjacent local minimum to $Z_1$
{along the path}.
We denote $-c_0 \ot c_1$ by $X$.
Then 
\begin{align*}
\Wt_X (Z_1) &= \wt_X (Z_1) + \varphi (a,b',c_0,c_1)\\
&
= \sum_{i=1}^l c_i - a + \frac12 (\min(a,b')-\min(c_0,c_1))\\
&
\geq \frac{c_1-a}{2} + \sum_{i=2}^l c_i.
\end{align*}
It follows that $\Wt_X (Z_1) \geq \frac{b-a}{2}$ for $l=1$, or
$\Wt_X (Z_1) \geq b - \frac{a}{2}$ for $l>1$, yielding
$\Wt_{Z_1} (X) = - \Wt_X (Z_1) < 0$.
Thus $X$ is not a floating minimum, hence fixed.
Then $\Wt_{Y_1} (Z_1) = \Wt_{Y_1} (X) + \Wt_X (Z_1) >0$.
\end{proof}
%

Now we establish the main result
(Theorem \ref{th:may21_1}) in the case of $\wt(p)=0$.
\begin{proposition}\label{prop:july24_4}
Let $p \in \mathcal{P}_{L,s}$ be a path satisfying $\wt (p) = 0$.
Then there exists $j$ 
such that $T_{s-1}^j (p)$ is equivalent to a ballot sequence.
\end{proposition}
\begin{proof}
Any $p \in \mathcal{P}_{L,s}$ has at least one maximal segment,
and if $\wt (p) = 0$ then $p$ has at least one fixed 
minimum (Lemma \ref{lem:jun21_1}).
Let $D$ be a maximal segment, and $Y_1$ be a lowest fixed minimum
in $p$.
In the second labeling scheme, $Y_1$ remains a lowest fixed 
minimum (Proposition \ref{rem:jul1_1}), 
hence a global minimum (Lemma \ref{prop:may18_3}) in any future time.
There exists $j$ such that $Y_1$ takes the spike form
in $T_{s-1}^j(p)$
(Proposition \ref{prop:july27_1}).
Then by Proposition \ref{prop:may25_5} the path $T_{s-1}^j(p)$
is equivalent to a ballot sequence.
\end{proof}

\subsection{The case of positive total weight}
\label{subsec:3_8}
In this subsection we establish 
Theorem \ref{th:may21_1} in the case of $\wt(p)>0$. 
%
First we consider the case in which $p$ has no fixed minima.
\begin{lemma}\label{lem:july24_5}
Any evolvable path of positive total weight
and having no fixed minima is equivalent to a ballot sequence.
\end{lemma}
\begin{proof}
Let $p = b_1 \ot \cdots \ot b_L$ be such a path.
Suppose $b_L \ot b_1$ is not a local minimum.
(If it is, consider $T_s(p)$ instead of $p$.)
By the same arguments and through the same
procedure in the proof of
Lemma \ref{lem:jun21_1}, we can obtain a sequence of types 
(\textbf{A},\textbf{B}, or \textbf{C}) of singular segments.
Since $p$ is of positive total weight 
and has no fixed minima, it must take
the form \textbf{A}$...$\textbf{A} or have no local minima at all.
Then any maximal segment in $p$ must be of length one
and take the form $s$.
Apply cyclic shifts appropriate times
to make one of them at the leftmost position, and
call the resulting path $p' = b'_1 (=s)\ot \cdots \ot b'_L$.
Then $p'$ is a ballot sequence because we have $w_0(p')=0$ and
$b'_k + b'_{k+1} > 0$,
hence $w_k(p') > 0$ for all $k \geq 1$.
\end{proof}
Next we consider the case in which $p$ has at least one fixed minimum.
Let $\mathcal{M}_{\rm FX}$ be
the set of all fixed minima in $p$.
We call $Y \in \mathcal{M}_{\rm FX}$ a 
\textbf{virtual lowest fixed minimum} if
$\Wt_{Y}(Y') \geq 0$
holds for any $Y' \in \mathcal{M}_{\rm FX}$.
Clearly such $Y$ exists, because
there exists $Y \in \mathcal{M}_{\rm FX}$ such that
$\Wt (Y') \geq \Wt (Y)$ holds for any $Y' \in \mathcal{M}_{\rm FX}$,
and we have (\ref{eq:jun21_5}).
By {virtue of} Proposition \ref{rem:jul1_1},
if $Y$ is a virtual lowest fixed minimum then it remains so
in any future time.
\begin{lemma}\label{lem:july24_6}
A virtual lowest fixed minimum is a virtual global minimum.
\end{lemma}
\begin{proof}
Let $p$ be a path of positive total weight with
at least one fixed minimum, 
$Y$ a virtual lowest fixed minimum in $p$.
We show that for any floating minimum $Z$ the relation
$\Wt_{Y}(Z) \geq 0$ holds.
We give a proof based on 
the same idea in that of Lemma \ref{prop:may18_3}.
Thus we can assume that the singular segment touching $Z$
does not touch any other local minima.
Then $Z$ is in one of the three cases in (\ref{eq:may18_1}).
For instance consider type \textbf{A} case.
(The other cases can be proved in a similar way.)
By the same argument in that proof,
we find there is an adjacent local minimum $X_1$ to the
left of $Z$ which yields $\Wt_{X_1} (Z) > 0$.
If $X_1$ is a floating minimum, then it must be in type \textbf{A} case
again, so we have $\Wt_{X_2} (X_1) > 0$ for the adjacent local minimum $X_2$
to the left of $X_1$.
By repeating this procedure we will arrive at some $X_n$
which is a fixed minimum.
Then $\Wt_{X_n} (Z) = \Wt_{X_n} (X_{n-1}) + \cdots + \Wt_{X_2} (X_1) 
+ \Wt_{X_1} (Z) > 0$.
This completes the proof since 
$\Wt_{Y}(Z) = \Wt_{Y}(X_n) + \Wt_{X_n} (Z) >0$.
\end{proof}

Now we establish the main result
(Theorem \ref{th:may21_1}) in the case of $\wt(p)>0$.
\begin{proposition}\label{prop:july24_7}
Let $p \in \mathcal{P}_{L,s}$ be a path satisfying $\wt (p) > 0$.
Then there exists $j$ 
such that $T_{s-1}^j (p)$ is equivalent to a ballot sequence.
\end{proposition}
\begin{proof}
If $p$ has no fixed minima, there is none to be done
(Lemma \ref{lem:july24_5}).
We consider the case where $p$ has at least one fixed minimum.
Then the proof proceeds in the same way as that of
Proposition \ref{prop:july24_4} if we replace
Lemma \ref{prop:may18_3} by
Lemma \ref{lem:july24_6}.
\end{proof}
\begin{example}
Consider 
\begin{math}
p=(21111)(22211)(21111)(21111)(22222)(22111),
\end{math}
which is the same path in Example \ref{ex:july23_1}.
Its time evolution is given as follows.
\begin{align*}
p &= 3 \ot \dot{1} \otb 3 \ot 3 \ot \dot{5} \otn 1 \, (\ot)\\
T_4(p) &= 3 \ot 3 \ot \dot{1} \otb 3 \ot 1 \ot \dot{5} \, (\otn)\\
T_4^2(p) &= \dot{5} \otn 5 \ot 3 \ot \dot{1} \otb 3 \ot \dot{1}\, (\ot)
\end{align*}
Here $\bullet$ denotes a floating minimum,
$\nabla$ a (virtual lowest) fixed minimum.
The last path is equivalent to a ballot sequence.
In fact we have
\begin{displaymath}
T_5^5 T_4^2(p) =(11111)(21111)(22211)(21111)(22211)(22222).
\end{displaymath}
\end{example}

\section{Proof in the Case of Even Cell Capacity}
\label{sec:4}
The main result of this paper
(Theorem \ref{th:may21_1})
is valid for any cell capacity $s$, regardless of whether it
is odd or even.
We shall not repeat all of the arguments in the previous section,
but rather discuss some of them that have
differences between the two cases.
To begin with
we consider the definition of local minimum/maximum.
Let $s(\geq 2)$ be an even integer.
Recall that 
given a $T_{s-1}$-evolvable path 
$p = b_1 \ot \cdots \ot b_L \in (B_s)^{\ot L}$
we can determine $v_i \in B_{s-1} \, (0 \leq i \leq L)$
uniquely as in (\ref{eq:jun14_2}).
\begin{definition}\label{def:july13_1}
$b_k \ot b_{k+1}$ is called a \textbf{local minimum}
(resp.~\textbf{local maximum}) if
$b_k \leq 0, \, b_{k+1} \geq 0$ and $v_k \leq -1$ 
(resp.~$b_k \geq 0, \, b_{k+1} \leq 0$ and $v_k \geq 1$ ).
A local minimum/maximum is also called an \textbf{extremum}.
\end{definition}
By this definition and Lemma \ref{lem:july8_1}
we obtain the following results immediately.
\begin{proposition}\label{prop:july18_2}
It follows that:
\begin{enumerate}
\item
If $b_k \leq -2$ and $b_{k+1} \geq 0$ 
(resp.~$b_k \geq 2$ and $b_{k+1} \leq 0$ )
then $b_k \ot b_{k+1}$ is a local minimum (resp. maximum).
\item If $b_{k+1} = 0$ then $b_k \ot b_{k+1}$ is an extremum.
\item In a level $0$ segment local minima and local maxima
alternate successively.
\item Any level $0$ segment is left singular.
\item $b_k \ot b_{k+1}$ and $b_{k+1} \ot b_{k+2}$ cannot be
local minima simultaneously.
\end{enumerate}
\end{proposition}
\begin{proof}
(1) The condition $b_k \leq -2$ implies $v_k \leq -1$, and
$b_k \geq 2$ does $v_k \geq 1$.
(2) In the case of $b_k=0$ it is clear by definition.
If $b_k \ne 0$ then it follows from Item 1.
(3) Follows from Item 2 and by Lemma \ref{lem:july8_1}.
(4) Follows from Item 1.
(5) If it can, $b_{k+1}=0$.
Then $v_k \leq -1$ implies $v_{k+1} \geq 1$ by Lemma \ref{lem:july8_1}.
\end{proof}
A generalization of Item 3 is given as follows.
\begin{proposition}
Local minima and local maxima alternate along a path.
\end{proposition}
\begin{proof}
Let $b_j \ot b_{j+1} \, (X)$ be a local minimum.
Among $b_{j+2},b_{j+3},\ldots$ let $b_{k+1}$ be the leftmost element
that satisfies $b_{k+1} \leq 0$.
Then $b_k \ot b_{k+1} \, (Y)$ is a local maximum and there {are}
no extrema between $X$ and $Y$.
Then among $b_{k+2},b_{k+3},\ldots$ let $b_{l+1}$ be the leftmost element
that satisfies $b_{l+1} \geq 0$.
Then $b_l \ot b_{l+1} \, (Z)$ is a local minimum and there are
no extrema between $Y$ and $Z$.
\end{proof}
\begin{coro}\label{prop:july13_2}
The assertion in Proposition \ref{prop:may30_1} also
holds in the case of even $s$.
\end{coro}
\begin{example}\label{ex:july10_1}
Consider the following path $p \in (B_4)^{\ot 8}$
\begin{displaymath}
p = 0  \ot 0 \ot -2 \ot 0 \ot 0 \ot -2 \ot 0 \ot 4.
\end{displaymath}
In this path there are three local maxima,
$b_8 \ot b_1(=4 \ot 0) , b_2 \ot b_3 (= 0 \ot -2)$, 
and $b_4 \ot b_5 (= 0 \ot 0)$.
The part $b_1 \ot b_2 (= 0 \ot 0)$ is a level $0$ singular
segment and a floating minimum.
The parts $b_3 \ot b_4, b_6 \ot b_7 (= -2 \ot 0)$ are fixed minima.
\end{example}
\begin{proposition}\label{prop:july11_1}
The assertion in Proposition \ref{lem:may16_1} also
holds in the case of even $s$.
\end{proposition}
\begin{proof}
Suppose $b_k \ot b_{k+1}$ is a local minimum, i.~e.~$b_k \leq 0,\, b_{k+1}\geq 0$
and $v_k \leq -1$.
In addition to (\ref{eq:july9_1})
\begin{displaymath}
\batten{v_{k-1}}{b_k}{b'_k}{v_k} \!\!\! \batten{}{b_{k+1}}{b'_{k+1}}{v_{k+1}}\,
\batten{}{b_{k+2}}{b'_{k+2}}{v_{k+2}}
\end{displaymath}
we use part of the diagram for the
next time step.
\begin{displaymath}
\batten{}{b'_k}{}{v'_k} \!\!\! 
\batten{}{b'_{k+1}}{}{v'_{k+1}}\,
\batten{}{b'_{k+2}}{}{} \\
\end{displaymath}
(1) Suppose $v_k + b_{k+1} \leq -1$.
Then by Lemma \ref{lem:july8_1}
we have $b'_{k+1} = v_k - 1 \leq -2$ and
$b'_{k+2} \geq v_{k+1} - 1 = b_{k+1} \geq 0$.
Hence $b'_{k+1} \ot b'_{k+2}$ is a local minimum
by Item 1 of Proposition \ref{prop:july18_2}.
(2) Suppose $v_k + b_{k+1} \geq 1$.
Then
we have $b'_{k+1} = v_k + 1 \leq 0$ and
$b'_{k+2} \geq v_{k+1} - 1 = b_{k+1}-2 \geq -v_k -1 \geq 0$.
It follows that if $b'_{k+1} \leq -2$
then $b'_{k+1} \ot b'_{k+2}$ is a local minimum, 
{otherwise} either
$b'_{k+1} \ot b'_{k+2}$  or $b'_{k} \ot b'_{k+1}$ 
is a local minimum according to $v'_k \geq 1$ or $v'_k \leq -1$. 
\end{proof}
\begin{coro}\label{prop:july11_2}
Suppose both
$b_{k} \ot b_{k+1}$ and
$b'_{k} \ot b'_{k+1}$
are local minima.
Then $b_{k} = -2$, $b_{k+1} \geq 2$, $v_k = -1$, 
and $b'_{k+1} =0$.
\end{coro}
\begin{proof}
{}From the proof of the proposition, 
we have $b'_{k+1} =0$.
Since $b_{k} \ot b_{k+1}$ is a local minimum, it requires
$v_k = -1$, hence $b_k = 0$ or $-2$.
Then, since $b'_{k+1} = v_k + 1$ implies
$v_k + b_{k+1} \geq 1$ (Lemma \ref{lem:july8_1})
we have $b_{k+1} \geq 2$.
Suppose $b_k=0$.
Then, since $b_{k} = v_k + 1$ implies
$v_{k-1}+b_k \geq 1$ we have
$b_k' = v_{k-1} + 1 \geq 2$.
This contradicts to our assumption $b_k' \leq 0$.
Hence $b_k = -2$.
\end{proof}
This corollary characterizes the cases where the local minimum
does not \textit{shift} to the right.
\begin{example}\label{ex:july11_4}
Consider the time evolution of
the path in Example \ref{ex:july10_1}.
We use $\dot{2},\dot{4},\ldots$ instead of $-2,-4$.
\begin{align*}
p &= 0  \stackrel{\bullet}{\ot} 0 \ot \dot{2} \stackrel{\vee}{\ot} 0 \ot 0 
\ot \dot{2} \stackrel{\nabla}{\ot} 0 \ot 4 \,(\ot) \\
T_3(p) &= 4  \ot \dot{2} \stackrel{\bullet}{\ot} 0 \ot \dot{2}
 \stackrel{\vee}{\ot} 2 \ot \dot{2} \ot \dot{2} \stackrel{\nabla}{\ot} 2 
\,(\ot) \\
T_3^2(p) &= 2  \ot 4 \ot \dot{4} \stackrel{\vee}{\ot} 0 \ot 
0 \stackrel{\bullet}{\ot} 0 \ot \dot{2} \stackrel{\nabla}{\ot} 0 \,(\ot) \\
T_3^3(p) &= 2  \ot 2 \ot 2 \ot \dot{4} \stackrel{\vee}{\ot}
2 \ot \dot{2} \stackrel{\bullet}{\ot} 0 \ot \dot{2} 
\,(\stackrel{\nabla}{\ot})\\
T_3^4(p) &= 0  \stackrel{\bullet}{\ot} 2 \ot 2 \ot 0 \ot
\dot{4} \stackrel{\vee}{\ot} 4 \ot \dot{4} \stackrel{\nabla}{\ot} 0
\,(\ot) \\
T_3^5(p) &= 2 \ot 0 \stackrel{\bullet}{\ot} 
2 \ot 2 \ot \dot{2} \ot \dot{2} \stackrel{\vee}{\ot}
2 \ot \dot{4} \,(\stackrel{\nabla}{\ot}) \\
T_3^6(p) &= \dot{4} \stackrel{\nabla}{\ot} 4 \ot 0
\stackrel{\bullet}{\ot} 2 \ot 0 \ot \dot{2} 
\stackrel{\vee}{\ot} 0 \ot 0 \,(\ot)
\end{align*}
Here $\bullet$ denotes a floating minimum,
$\vee$ a fixed minimum (not lowest), and
$\nabla$ a lowest fixed minimum.
\end{example}

In subsec.~\ref{subsec:3_4} we classified the moves of
the local minima into three types.
In the case of even $s$ we classify them into two types.
Recall the diagram in Proposition \ref{prop:july11_1}
where $b_k \ot b_{k+1}$ is assumed to be a local minimum.
If $v_k + b_{k+1} \geq 1$ {then} we call the move \textbf{type I},
while if $v_k + b_{k+1} \leq 1$ {then} call it \textbf{type II}.
Although the type I move includes two cases according to
whether the local minimum shifts or not,
we do not distinguish them.
\begin{remark}\label{rem:july24_2}
Consider the configurations of local minima as in Remark \ref{rem:july24_1}.
For any $v_k \in \{-1,-3,\ldots , -s+1 \}$
we can choose $b_k = v_k \pm 1$ and $b_{k+1} \in \{ 0,2, \ldots, s \}$.
So there are $\frac{s(s+2)}{2}$ configurations in total.
Half of them are classified into type I, and the others are into type II.
\end{remark}
\begin{lemma}\label{prop:july11_3}
The assertion in Lemma \ref{lem:may16_3} also
holds in the case of even $s$.
\end{lemma}
\begin{proof}
We first recall the part for type I move in
the proof of Lemma \ref{lem:may16_3}.
It does not use the fact that the local minimum
shifts to the right.
So it is also valid in the present case,
regardless of whether the local minimum shifts or not.
Next we consider the part for type II.
Although it uses Lemma \ref{lem:jun14_1} which are not applicable
to level $0$ segments, 
we have Item 4 of Proposition \ref{prop:july18_2} instead.
So it is also valid.
Finally, the part for type III is now irrelevant.
\end{proof}
\begin{example}\label{ex:july11_5}
Consider the transition from $T_3^3(p)$ to $T_3^4(p)$ 
in Example \ref{ex:july11_4}.
\begin{displaymath}
\batten{-1}{2)}{(0}{1}\!\!\!\!\!
\batten{}{2}{2)}{1}\!\!\!\!\!
\batten{}{2}{2}{1}\!\!\!\!\!
\batten{}{(-4}{0}{-3}\!
\batten{}{2)}{(-4}{3}\!\!\!\!\!
\batten{}{(-2}{4)}{-3}\!
\batten{}{0)}{(-4}{1}\!\!\!\!\!
\batten{}{(-2}{0)}{-1}
\end{displaymath}
Since $b_6'= 4 = v_5 +1$, we see that $b_4 \ot b_5 (= -4 \ot 2)$ is
a fixed minimum by Lemma \ref{prop:july11_3}.
In the same way $b_8'= 0 = v_7 -1$ means
$b_6 \ot b_7( = -2 \ot 0)$ is
a floating minimum, and
$b_8= -2 = v_8 -1$ means
$b_8 \ot b_1( = -2 \ot 2)$ is
a fixed minimum.
\end{example}
\begin{proposition}\label{prop:july13_7}
The assertions in Lemmas \ref{prop:apr12_1} and \ref{prop:may23_2}
also hold in the case of even $s$
(except the part for type III which is now irrelevant).
\end{proposition}
\begin{proof}
{Regarding} Lemma \ref{prop:apr12_1}, {we find}
it suffices to show that the claim is valid in the case where
$X$ or $Y$ (or both of them) is of type I and does not
shift to the right, {because only this case is new}.
First we suppose $X$ is so but $Y$ is not.
Then by Corollary \ref{prop:july11_2}
the situation is given by (\ref{eq:may25_2})
except the $(v_j=)0$ is now replaced by $-1$.
Hence we obtain
$-1 + \wt_X(Y)  + b_{k+1} = \wt'_X(Y) + v_{k+1}$.
But this is equivalent to
\begin{equation}\label{eq:july13_8}
v_j + \wt_X(Y)  + b_{k+1} = b'_{j+1} + \wt'_X(Y) + v_{k+1}
\end{equation}
because we have $v_j=-1$ and $b'_{j+1}=0$.
This yields
$\Delta \wt_X(Y) = v_j -b'_{j+1} - (v_{k+1}-b_{k+1})$,
hence the desired result.
Next we suppose $Y$ is of that type but $X$ is not.
Then 
the situation is given by (\ref{eq:july13_6})
except the $(v_k=)0$ is replaced by $-1$.
Hence we obtain
$v_j + \wt_X(Y)  = b'_{j+1} + \wt'_X(Y)  + (-1)$.
But this is also equivalent to (\ref{eq:july13_8})
because we have $v_{k+1}-b_{k+1} = -1$ by Lemma \ref{lem:july8_1}
and Corollary \ref{prop:july11_2}.
The remaining case and Lemma \ref{prop:may23_2}
can be treated in the same way.
\end{proof}
\begin{proposition}\label{prop:july13_9}
The assertions in Propositions \ref{prop:may23_1} 
and \ref{rem:jul1_1} also
hold in the case of even $s$.
\end{proposition}
\begin{proof}
Let $\varphi (b_m',b_{m+1}',b_k,b_{k+1})$ be the same one
in the proof of Proposition \ref{prop:may23_1},
and suppose $Y$ takes type I move and does not
shift to the right.
Then we have 
$\varphi (b_m',b_{m+1}',b_k,b_{k+1}) = 
\frac12 (\min(-b_k',b_{k+1}')- \min(-b_k,b_{k+1})) = -1$
by Corollary \ref{prop:july11_2}.
Thus the proof is also valid in the present case.
Proposition \ref{rem:jul1_1} is based on
Lemma \ref{lem:jun13_1} that holds regardless of whether
$s$ is odd or even, hence it is valid too.
\end{proof}

Now we confirmed the validity of the arguments in
subsections \ref{subsec:3_1} - \ref{subsec:3_6}
in the case of even $s$.
Since the arguments in
subsections \ref{subsec:3_7} and \ref{subsec:3_8}
are valid in both cases,
the main result of this paper
(Theorem \ref{th:may21_1}) is now established
for all $s$.

\begin{example}\label{ex:july18_5}
Consider
$p=(2211)(2211)(2221)(2211)(2211)(2221)(2211)(1111)$,
that is the same path in Examples \ref{ex:july10_1} and \ref{ex:july11_4}.
In Example \ref{ex:july11_4} 
the maximal segment overlaps with
the lowest fixed minimum to make a spike in $T_3^6(p)$.
Thus it becomes a ballot sequence via cyclic shifts.
Indeed we have
\begin{displaymath}
T_4^7 T_3^6(p) = 
(1111)(2211)(2111)(2211)(2221)(2211)(2211)(2222).
\end{displaymath}

\end{example}

\section{Discussions}\label{sec:5}
In this paper we introduced a family of integrable periodic cellular automata (CA) associated with 
the crystal basis of $s$-fold symmetric tensor representation of $U_q(A^{(1)}_1)$ for any positive integer $s$.
This family of CA has the following physical meaning.
In the case of $s=1$ this CA is known to be an ultra-discrete limit of 
discrete KdV and discrete Toda equations with periodic boundary conditions.
We consider $s > 1$ case.
In the infinite system size limit this CA falls into a special case of the non-periodic system in ref.~\citen{HHIKTT}.
It is an ultra-discrete limit of non-autonomous discrete KP (ndKP) equation.
One can consider certain tau functions satisfying ndKP or a set of 
Hirota bilinear equations.
The ultra-discrete limit of
``cross ratios" of those tau functions coincide with elements of crystals
in the sense that they satisfy the intertwining relation under the combinatorial $R$ \cite{HHIKTT}.
Since this correspondence is based on spatially local properties it remains valid also in the periodic boundary case.

A remark on the conservations of the kink-like structures
is in order.
These structures are conserved under the time evolution $T_{s-1}$.
See, for example, Corollary \ref{cor:july10_2}.
Besides the cyclic shift $T_s$
one can also show that they are conserved under $T_{s+1}$.
However they are not conserved under $T_i \, (i \ne s,s\pm1)$ in general.
Since the known conserved quantities of generalized pBBS do not change their values under any $T_i$ \cite{KS},
it shows that the kink-like structures are independent of
those conserved quantities.

Finally, as a possible application we point out that our CA will be used in informatics
if one can find an appropriate interpretation of the sequences.
In this direction there is a study on the relation between 
ultra-discrete soliton equations and sorting algorithms \cite{NTS}.
Given a sequence of real numbers a sorting algorithm sorts the numbers in decreasing order.
We think that in some situation our ballot sequences can play a similar role as sequences in such order, and
then the algorithm of the quasi-cyclic shift will be an efficient sorting algorithm.
\section*{Acknowledgements}
The author thanks Atsuo Kuniba and Reiho Sakamoto for 
valuable discussions.

\appendix
\section{Proof of Proposition \ref{prop:july19_1}}\label{app:a}

\begin{proof}[Proof of Proposition \ref{prop:july19_1}]
Given $p = b_1 \ot \cdots \ot b_L \in \mathcal{P}_{L,2}$
with $\wt (p) = 0$ we let
$T_1(p) = b'_1 \ot \cdots \ot b'_L$.
There are four types of local minima:
The local minimum $-2 \ot 2$ (spike) is always fixed and 
$0 \ot 2$ is always floating.
The other types ($-2 \ot 0$ or $0 \ot 0$)
are either fixed or floating.
We claim that:
\begin{enumerate}
\item A global minimum of the form $-2 \ot 2$
does not coexist with that
of the form $-2 \ot 0, 0 \ot 0$ or $0 \ot 2$.
\item A fixed minimum of the form $-2 \ot 0$ or $0 \ot 0$ in $p$
moves to a fixed minimum of the form $-2 \ot 2$ in $T_1(p)$, 
and vice versa.
\end{enumerate}
In Items 2, the second labeling scheme is assumed.

Item 1 is proved as follows.
Let $X$ be a local minimum of the form $-2 \ot 2$,
and $Y$ that of the form $-2 \ot 0, 0 \ot 0$ or $0 \ot 2$.
Then (\ref{eq:jun21_4}) requires $\Wt_X(Y)$ to be an odd integer,
hence $\ne 0$.
Thus if $X$ is a global minimum then $Y$ is not, and vice versa.

Item 2 is proved as follows.
Let $b_k \ot b_{k+1} \, (X)$ be a fixed minimum in $p$.
(1) If $X$ takes the form
$-2 \ot 0$ or $0 \ot 0$, then
$b_{k+2} \ne -2$.
This leads to $b'_{k+1} \ot b'_{k+2} = -2 \ot 2$,
that is $X$ in $T_1(p)$.
(2) If $X$ takes the form
$-2 \ot 2$ then $b'_{k+1} =0$.
In the first labeling scheme this requires $X$ to take 
one of the forms
$0 \ot 2,-2 \ot 0$, or $0 \ot 0$ in $T_1(p)$.
The possibility of $0 \ot 2$ is excluded in the second scheme,
because $X$ is a fixed minimum.

By these claims we can prove the proposition as follows.
Since $p$ is of zero total weight, it has a fixed minimum
(Lemma \ref{lem:jun21_1}).
(1) Suppose a lowest fixed minimum takes the form $-2 \ot 2$.
Since it is a global minimum (Lemma \ref{prop:may18_3}),
the path $p$ is equivalent to a ballot sequence
(Proposition \ref{prop:may25_5}).
In this case we find from Item 1 that
every global minimum in $p$ takes the spike form.
Then by Item 2 none of the global minima in $T_1(p)$
takes this from.
Hence $T_1(p)$ is not equivalent to a ballot sequence
(Proposition \ref{prop:may25_5}).
(2) Suppose a lowest fixed minimum (also a global minimum)
takes the form $-2 \ot 0$ or $0 \ot 0$.
Then by Item 1 none of the global minima in $p$
takes the spike from.
Hence $p$ is not equivalent to a ballot sequence
(Proposition \ref{prop:may25_5}).
In this case we find from Item 2 that
a lowest fixed minimum in $T_1(p)$ takes the spike form.
Hence $T_1(p)$ is equivalent to a ballot sequence
(Proposition \ref{prop:may25_5}).
\end{proof}

\end{document}